\documentclass[11pt,reqno]{amsart}

\usepackage{geometry}                % See geometry.pdf to learn the layout options. There are lots.
\usepackage{amsmath,amssymb,amsthm}

\usepackage{amsmath,amssymb,amsthm}
 \usepackage[authoryear]{natbib}
\usepackage{graphicx}
\usepackage{comment}

\theoremstyle{plain}
\newtheorem{theorem}{Theorem}[section]	
\newtheorem{lemma}{Lemma}[section]

\theoremstyle{definition}

\newtheorem{remark}{Remark}[section]

\DeclareMathOperator{\rank}{rank}

\DeclareMathOperator{\tr}{tr}

\DeclareMathOperator{\PR}{\mathsf{P}}
\DeclareMathOperator{\K}{\mathsf{K}}
\DeclareMathOperator{\Q}{\mathsf{Q}}

\renewcommand{\qed}{\hfill{\tiny \ensuremath{\blacksquare} }}%

\newcommand{\Ep}{{\mathrm{E}}}

\renewcommand{\Pr}{{\mathrm{P}}}

\newcommand{\df}{\mathrm{df}}

\newcommand{\sh}{{d}}
\newcommand{\lasso}{\mathrm{lasso}}
\newcommand{\ridge}{\mathrm{ridge}}
\newcommand{\lava}{\mathrm{lava}}

\numberwithin{equation}{section}

\begin{document}

%%
%% The title of the paper goes here.  Edit to your title.
%%

\title{A lava attack on the recovery of sums of dense and sparse signals}\thanks{We are
grateful to Garry Chamberlain, Guido Imbens, Anna Mikusheva, Philippe Rigollet  for helpful discussions.}

%:
%%
%% Now edit the following to give your name and address:
%% 

\author{Victor Chernozhukov}
\address{Department of Economics, MIT,  Cambridge, MA 02139}
\email{vchern@mit.edu}
%\urladdr{www.math.sc.edu/$\sim$howard} % Delete if not wanted.

\author{Christian Hansen}
\address{Booth School of Business, University of Chicago, Chicago, IL 60637}
\email{Christian.Hansen@chicagobooth.edu}
%\urladdr{www.math.sc.edu/$\sim$howard} % Delete if not wanted.
 
\author{Yuan Liao}
\address{Department of Mathematics, University of Maryland, College Park, MD 20741}
\email{yuanliao@umd.edu}
%\urladdr{www.math.sc.edu/$\sim$howard} % Delete if not wanted.

%%
%% If there is another author uncomment and edit the following.
%%

%\author{Second Author}
%\address{Department of Mathematics, University of South Carolina,
%Columbia, SC 29208}
%\email{second@math.sc.edu}
%\urladdr{www.math.sc.edu/$\sim$second}

%%
%% If there are three of more authors they are added in the obvious
%% way. 
%%

%%%
%%% The following is for the abstract.  The abstract is optional and
%%% if not used just delete, or comment out, the following.
%%%

\begin{abstract}
Common high-dimensional methods for prediction rely on having either a sparse signal model, a model in which most parameters are zero and there are a small number of non-zero parameters that are large in magnitude, or a dense signal model, a model with no large parameters and very many small non-zero parameters.  We consider a generalization of these two basic models, termed here a ``sparse+dense" model, in which the signal is given by the sum of a sparse signal and a dense signal.  Such a structure poses problems for traditional sparse estimators, such as the lasso, and for traditional dense estimation methods, such as ridge estimation.  We propose a new penalization-based  method, called lava, which is computationally efficient. With suitable choices of  penalty parameters, the proposed method strictly dominates both lasso and ridge.  We derive analytic expressions for the finite-sample risk function of the lava estimator in the Gaussian sequence model.  We also provide an  deviation bound for the prediction risk in the Gaussian regression model with fixed design. In both cases, we provide Stein's unbiased estimator for lava's prediction risk. A simulation example compares the performance of lava to lasso, ridge, and elastic net in a regression example using feasible, data-dependent penalty parameters and illustrates lava's improved performance relative to these benchmarks.
 
 %, akin to that given for lasso by \cite{Bickeletal}. 

%Common high-dimensional methods for prediction rely on either sparse or dense signal models.  We propose to generalize these two basic models into a ``sparse+dense'' model, which is given by the sum of a sparse signal, having relatively few large coefficients, and a dense signal, having possibly very many small coefficients.   Such signals render traditional sparse or dense estimation methods, such as lasso or ridge, not  optimally suitable for prediction and other estimation purposes.  We propose a new estimation method, called lava, based on penalization that combines features of the two methods, and, with suitable choices of  penalty levels, strictly dominates both methods when the signal is ``sparse+dense''.  We provide analytical expressions for the finite-sample risk function of the lava estimator in the Gaussian sequence model and a fixed design regression model with Gaussian errors, and exhibit examples where lava significantly outperforms both lasso and ridge.  We also provide Stein's  unbiased estimators for lava's  risk. We  conclude with a simulation study that examines the performance of lava against lasso and ridge with practical choices of tuning parameters.

\end{abstract}

  \maketitle

%%
%% LaTeX can automatically make a table of contents.  This is done by
%% uncommenting the following:
%%

%\tableofcontents

%%
%%  To enter text is easy.  Just type it.  A blank line starts a new
%%  paragraph. 
%%

 
 \section{Introduction}
 
Many recently proposed high-dimensional modeling techniques  build upon the fundamental assumption of sparsity.  Under sparsity, we can approximate a high-dimensional signal or parameter by a sparse vector that has a relatively small number of non-zero components.  Various $\ell_1$-based penalization methods, such as the lasso and soft-thresholding, have been proposed for signal recovery,  prediction, and parameter estimation within a sparse signal framwork. See \cite{frank1993statistical},  \cite{donoho1995adapting}, \cite{tibshirani96}, \cite{fan2001variable}, \cite{efron2004least}, \cite{zou2005regularization}, \cite{zhao2006model}, \cite{yuan2006model},  \cite{bunea2007sparsity}, \cite{candes2007dantzig}, \cite{fan2008sure}, \cite{Bickeletal}, \cite{meinshausen2009lasso}, \cite{wainwright2009sharp}, \cite{bunea2010spades}, \cite{zhang2010nearly}, \cite{loh2013regularized}, and others.   By virtue of being based on $\ell_1$-penalized optimization problems, these methods produce sparse solutions in which many estimated model parameters are set exactly to zero.  

Another commonly used shrinkage method is ridge estimation. Ridge estimation differs from the aforementioned $\ell_1$-penalized approaches in that it does not produce a sparse solution but instead provides a solution in which all model parameters are estimated to be non-zero. Ridge estimation is thus particularly suitable when the model's parameters or unknown signals contain many very small components, i.e. when the model is dense.  See, e.g., \cite{ hsu2014random}.  Ridge estimation tends to work better than sparse methods whenever a signal is dense in such a way that it can not be well-approximated by a sparse signal. %Asymptotic analyses of ridge estimations have been studied by, e.g., \cite{ hsu2014random}.   %See \cite{antoniadis2001regularization} for a general discussion of shrinkage estimations. 

%   A leading example is   non-parametric regressions using high-dimensional sieves: the nonparametric regression function is often represented as a linear combination of many sieve basis. When the unknown function is not sufficiently smooth, the sieve coefficients can decay slowly so that  using a few large coefficients to represent the sieve coefficient vector an produce noticeable prediction loss. 
%In addition,  the sieve coefficients that are sparse-representable (or ``compressible") under a given set of basis functions may no longer be sparse-representable under a different set of basis functions. So choosing the ``right basis" is a crucial requirement. However, this has been often taken as a given condition, and there is no practical clue on this perspective.   

In practice, we may face environments that have signals or parameters which are neither dense nor sparse.  The main results of this paper provide a model that is appropriate for this environment and a corresponding estimation method with good estimation and prediction properties. Specifically, we consider models  where the signal or parameter, $\theta$, is given by the superposition of sparse and dense signals:
\begin{equation}\label{model}
\theta=\underbrace{\beta}_{\text{dense part}}+\underbrace{\delta}_{\text{sparse part}}.
\end{equation}
Here, $\delta$ is a sparse vector that has a relatively small number of large entries, and $\beta$ is a dense vector having possibly very many small, non-zero entries.  Traditional sparse estimation methods, such as lasso, and traditional dense estimation methods, such as ridge, are tailor-made to handle respectively sparse signals and dense signals.  However, the model for $\theta$ given above is ``sparse+dense" and cannot be well-approximated by either a ``dense only" or ``sparse only" model.  Thus, traditional methods designed for either sparse or dense settings are not optimal within the present context.

Motivated by this signal structure,  we propose a new estimation method, called ``lava".   Let $\ell(\mathrm{data}, \theta)$ be a general statistical loss function that depends on unknown parameter $\theta$, and let $p$ be the dimension of $\theta$.  To estimate $\theta$, we propose the ``lava" estimator given by
\begin{equation}\label{define lava 1a}
\widehat\theta_{\lava}= \widehat\beta + \widehat\delta 
\end{equation}
where  $\widehat\beta $ and $ \widehat\delta$ solve the following penalized optimization problem:
\begin{equation}\label{define lava 1b}
(\widehat\beta,\widehat\delta)=\arg\min_{(\beta',\delta')'\in\mathbb{R}^{2p} } \Big \{ \ell(\mathrm{data}, \beta+\delta)+\lambda_2\|\beta\|_2^2+\lambda_1\|\delta\|_1 \Big \}.
\end{equation}
In the formulation of the problem, $\lambda_2$ and  $\lambda_1$ are tuning parameters corresponding to the $\ell_2$- and $\ell_1$- penalties which are respectively applied to the dense part of the parameter, $\beta$, and the sparse part of the parameter, $\delta$.  The resulting estimator is then the sum of a dense and a sparse estimator.   
Note that the separate identification of $\beta$ and $\delta$ is not required in (\ref{model}), and the lava estimator  is designed to  automatically recover the combination $\widehat\beta+\widehat\delta$ that leads to the optimal prediction of $\beta + \delta$. Moreover, under  standard conditions for $\ell_1$-optimization, the lava solution exists and is unique. In naming the proposal  ``lava", we emphasize that it is able, or at least aims, to capture or wipe out both sparse and dense signals.

The lava estimator admits the lasso and ridge shrinkage methods as two extreme cases by respectively  setting either $\lambda_2=\infty$ or $\lambda_1=\infty$.\footnote{With $\lambda_1 = \infty$ or $\lambda_2 = \infty$, we set $\lambda_1\|\delta\|_1 = 0$ when $\delta = 0$ or $\lambda_2\|\beta\|_2^2 = 0$ when $\beta = 0$ so the problem is well-defined.} In fact, it continuously connects the two shrinkage functions in a way that guarantees it will never produce a sparse solution when $\lambda_2<\infty$.  Of course, sparsity is not a requirement for making good predictions. By construction, lava's prediction risk is less than or equal to the prediction risk of the lasso and ridge methods with oracle choice of penalty levels for ridge, lasso, and lava; see Figure \ref{f1}.  
%With ``oracle" penalty levels or penalty levels chosen by cross-validation, 
Lava also tends to perform no worse than, and often performs significantly better than, ridge or lasso with penalty levels chosen by cross-validation; see Figures \ref{f4} and \ref{f5}.  

Note that our proposal is rather different from the elastic net method, which also uses a combination  of $\ell_1$ and $\ell_2$ penalization.  The elastic net penalty function is $\theta \mapsto \lambda_2 \|\theta\|^2_2 + \lambda_1 \|\theta\|_1$, and thus the elastic net also includes lasso and ridge as extreme cases corresponding to $\lambda_2 = 0$ and $\lambda_1 = 0$ respectively.  In sharp contrast to the lava method, the elastic net does \textit{not} split $\theta$ into a sparse and a dense part and will produce a sparse solution as long as $\lambda_1 > 0$.  Consequently, the elastic net method can be thought of as a sparsity-based method with additional shrinkage by ridge.  The elastic net processes data very differently from lava (see Figure 2 below) and consequently has very different prediction risk behavior (see Figure 1 below).

We also consider the post-lava estimator which refits the sparse part of the model:
\begin{equation}\label{define post-lava 1}
\widehat\theta_{\mathrm{post}\text{-}\lava}= \widehat\beta + \widetilde \delta,
\end{equation}
where $ \widetilde \delta$ solves the following penalized optimization problem:
\begin{equation}\label{define post-lava 2}
\widetilde \delta=\arg\min_{\delta \in\mathbb{R}^{p} } \Big \{ \ell(\mathrm{data}, \widehat \beta + \delta): \delta_j = 0, \text{ if } \widehat \delta_j =0 \Big \}.
\end{equation}
This estimator removes the shrinkage bias induced by using the $\ell_1$ penalty in estimation of the sparse part of the signal.  Removing this bias sometimes results in further improvements of lava's risk properties.

  \begin{figure}[htbp]
\begin{center}
  \includegraphics[width=10cm]{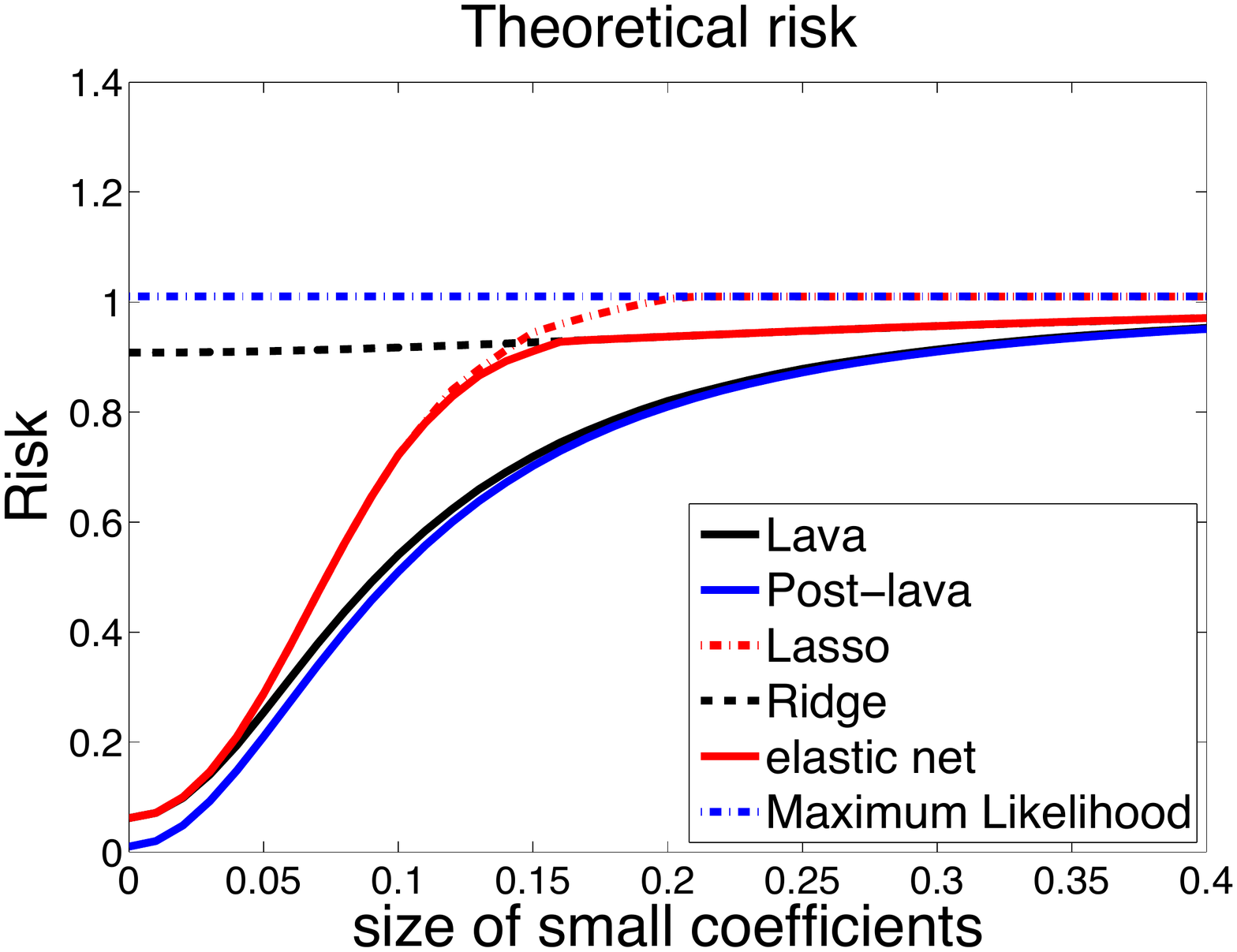}
 \includegraphics[width=10cm]{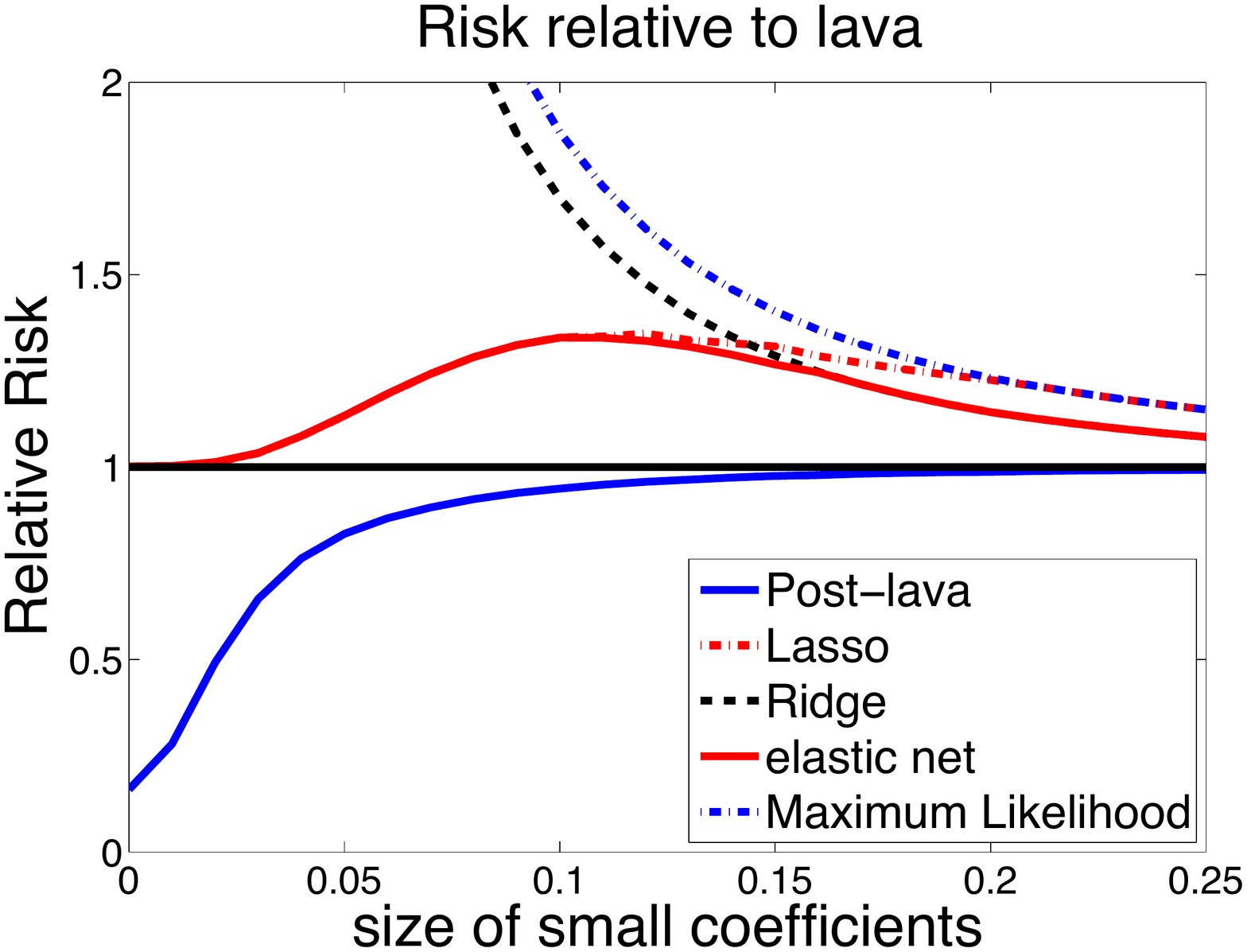}
\label{f1}
\end{center}
\caption{\footnotesize Exact risk and relative risk functions of lava, post-lava, ridge, lasso,  elastic net, and maximum likelihood in the Gaussian sequence model with ``sparse+dense'' signal structure, using the oracle (risk minimizing) choices of penalty levels.  See Section 2.5 for the description of the model. The size of ``small coefficients" is shown on the horizontal axis. The size of these coefficients directly corresponds to the size of the ``dense part" of  the signal, with zero corresponding to the exactly sparse case.   Relative risk plots the ratio of the risk of each estimator to the lava risk. 
%Relative risk plots $\mathrm{R}(\theta, \widehat \theta_e)/\mathrm{R}(\theta, \widehat \theta_{\lava})$.  
Note that the relative risk plot is over a smaller set of sizes of small coefficients to accentuate comparisons over the region where there are the most interesting differences between the estimators.
}
 
\end{figure}

We provide several theoretical and computational results about lava in this paper. First, we provide analytic expressions for the finite-sample risk function of the lava estimator as well as for other methods in the Gaussian sequence model and in a fixed design regression model with Gaussian errors.  Within this context, we exhibit ``sparse+dense" examples where lava significantly outperforms both lasso and ridge.  Stein's unbiased risk estimation plays a central role in our theoretical analysis, and we thus derive Stein's unbiased risk estimator (SURE) for lava.  We also characterize lava's ``Efron's'' degrees of freedom (\cite{efron2004estimation}). Second, we give deviation bounds for the prediction risk of the lava estimator in regression models akin to those derived by  \cite{Bickeletal} for lasso.  Third, we illustrate lava's performance relative to lasso, ridge, and elastic net through simulation experiments using penalty levels chosen via either minimizing the SURE or by k-fold cross-validation for all estimators.  In our simulations, lava outperforms lasso and ridge in terms of prediction error over a wide range of regression models with coefficients that vary from having a rather sparse structure to having a very dense structure.  When the model is very sparse, lava performs as well as lasso and outperforms ridge substantially.  As the model becomes more dense in the sense of having the size of the ``many small coefficients" increase, lava outperforms lasso and performs just as well as ridge.   This is consistent with our theoretical results.
%An interested reader now may take a look at  Figures 2 and 3 (dealing with Gaussian regression models) that present comparisons of risk of lava and of other methods in these cases.

We conclude the introduction  by noting that our proposed approach complements other recent approaches
to structured sparsity problems such as those considered in fused sparsity estimation (\cite{tibshirani2005sparsity} and \cite{dalalyan2012fused}) and structured matrix estimation problems (\cite{candes2011robust}, \cite{chandrasekaran2011rank}, \cite{POET}, and \cite{klopp2014robust}).  The latter line of research
studied estimation of matrices that can be written as  low rank plus sparse matrices.  Our new results are related to but are sharply different from this latter line of work since our focus is on regression problems.  Specifically, our chief objects of interest are regression coefficients along with the associated regression function and predictions of the outcome variable.  Thus, the target statistical applications of our developed methods include prediction, classification, curve-fitting, and supervised learning.  Another noteworthy point is that it is impossible to recover the ``dense" and ``sparse" components  separately within our framework; instead, we recover the sum of the two components.  By contrast, it is possible to recover the low-rank component of the matrix separately from the sparse part in some of the structured matrix estimation problems.  This distinction serves to highlight the difference between structured matrix estimation problems and the framework discussed in this paper.  Due to these differences, the mathematical side of our analysis needs to address a completely different set of issues than are addressed in the aforementioned structured matrix estimation problems.

%Our new results are related to but are clearly sharply different from this latter  line of work, since we focus on regression problems, with the target being the regression coefficients, which we then use to build better predictions for outcomes and the regression function.   Thus our  efforts are geared towards the regression problems, with the target statistical applications including prediction, classification, curve-fitting, supervised learning etc. Another noteworthy point, which highlights the distinction, is that in our framework it is impossible to recover the ``dense" and ``sparse" components  separately; so instead we recover the  sum of the two components.  By contrast, in some of the structured matrix estimation problems it is possible to recover the low-rank component of the matrix, separately from the sparse part.  Thus mathematical side of our analysis has to address a completely different set of issues.

We organize the rest of the paper as follows. Section 2 defines the lava shrinkage estimator in a canonical Gaussian sequence model, and derives its theoretical risk function. Section 3 defines and analyzes the lava estimator in the regression model. Section 4 provides computational examples, and Section 5 concludes.  We give all proofs in the appendix.

\textbf{Notation.}  The notation $a_n \lesssim b_n$ means that $a_n \leq C b_n$ for all $n$, for some constant $C$
that does not depend on $n$. The $\ell_{2}$ and $\ell_{1}$ norms are denoted by
$\|\cdot\|_2$ (or simply $\|\cdot\|$) and $\| \cdot \|_{1}$, respectively.  The $\ell_{0}$-``norm", $\|\cdot\|_0$, denotes the number of non-zero components of a vector, and the $\|.\|_{\infty}$ norm denotes a vector's maximum absolute element.  When applied to a matrix, $\|\cdot\|$ denotes the operator norm.
We use the notation $a \vee b = \max( a, b)$ and $a \wedge b = \min(a , b)$.  We use $x'$ to denote the transpose of a column vector $x$.

\section{The lava  estimator in a canonical model}
 
\subsection{The one dimensional case}
Consider the simple problem where a scalar random variable is given by
$$Z = \theta+ \epsilon,\quad \epsilon \sim N(0,\sigma^2).$$
We observe a realization $z$ of $Z$ and wish to estimate $\theta$.  Estimation will often involve the use of regularization or shrinkage
via penalization to process input $z$ into output $\sh(z)$, where the map
$z \mapsto \sh(z)$ is commonly referred to as the shrinkage (or decision) function.  A generic shrinkage
estimator then takes the form $\widehat \theta = \sh(Z)$. 

The commonly used lasso method uses $\ell_1$-penalization and gives rise to the lasso or soft-thresholding shrinkage 
function:
$$
\sh_{\lasso}(z)=\ \arg\min_{\theta \in \mathbb{R}} \Big \{ (z-\theta)^2+\lambda_l|\theta| \Big \} = (|z|-\lambda_l/2)_+\text{sign}(z),
$$
where $y_+ := \max(y,0)$ and $\lambda_l \geq 0$ is a penalty level. The use of the $\ell_2$-penalty in place of the $\ell_1$ penalty yields the ridge shrinkage function:
$$
\sh_{\ridge}(z)=\ \arg\min_{\theta \in \mathbb{R}} \ \Big \{ (z-\theta)^2+\lambda_r |\theta|^2  \Big\} =\frac{z}{1+\lambda_r},
$$
where $\lambda_r \geq 0$ is a penalty level. The lasso and ridge estimators then take the form $$\widehat \theta_{\lasso} = \sh_{\lasso} (Z), \quad 
\widehat \theta_{\ridge} = \sh_{\ridge} (Z).$$ 
Other commonly used shrinkage methods include the elastic-net (\cite{zou2005regularization}), which uses $\theta \mapsto \lambda_2|\theta|^2+\lambda_1|\theta|$ as the penalty function; hard-thresholding; and the SCAD (\cite{fan2001variable}), which uses a non-concave penalty function.

Motivated by points made in the introduction, we proceed differently. We decompose the signal into two components $$\theta=\beta+\delta,$$ 
and use the different penalty functions -- the $\ell_2$ and $\ell_1$ -- for each component in order to predict $\theta$ better.
%We proceed in a different fashion from these approaches by supposing that the signal has two components $$\theta=\beta+\delta.$$  We then use different penalty functions -- the $\ell_2$-penalty for $\beta$ and the $\ell_1$-penalty for $\delta$ -- where the choice is motivated by wishing to consider a signal that consists of both a sparse and a dense component. 
We thus consider the penalty function 
$$(\beta, \delta) \mapsto \lambda_2|\beta|^2+\lambda_1|\delta|,$$ 
and introduce the ``lava" shrinkage function  $z \mapsto \sh_{\lava}(z)$ defined by
\begin{eqnarray}\label{eq1.1.1}
\sh_{\lava}(z) := \sh_2(z) + \sh_1(z), \quad 
\end{eqnarray}
where  $\sh_1(z)$ and $\sh_2(z)$ solve the following penalized prediction problem:
\begin{eqnarray}\label{eq1.1}
(\sh_2(z),\sh_1(z)):=\arg\min_{(\beta,\delta) \in \mathbb{R}^2}\ \Big \{[z- \beta-\delta]^2+\lambda_2|\beta|^2+\lambda_1|\delta| \Big \}.
\end{eqnarray}
Although the decomposition $\theta=\beta+\delta$ is not unique, the optimization problem (\ref{eq1.1})
has a unique solution for any given $(\lambda_1, \lambda_2)$. 
The proposal thus defines the lava estimator of $\theta$:
$$\widehat\theta_{\lava} = \sh_{\lava}(Z).$$   

For large signals such that $|z|>\lambda_1/(2k)$, lava has the same bias as the lasso.  This bias can be removed through the use of the post-lava estimator$$
\widehat\theta_{\mathrm{post}-\lava} = \sh_{\mathrm{post}-\lava}(Z),$$
 where $\sh_{\mathrm{post}-\lava}(z) := \sh_2(z) + \tilde \sh_1(z),$ 
 and $\tilde \sh_1(z)$ solves the following penalized prediction problem:
\begin{eqnarray}\label{eq1.1.pl}
\tilde \sh_1(z):=\arg\min_{ \delta \in \mathbb{R}}\Big\{ [z-\sh_2(z)-\delta]^2:  \delta =0 \text{ if }  \sh_1(z) =0 \Big \}.
\end{eqnarray}
The removal of this bias will result in improved risk performance relative to the original estimator in some contexts.

From the Karush-Kuhn-Tucker conditions, we obtain the  explicit solution to (\ref{eq1.1.1}).
\begin{lemma}\label{l2.1}
For given penalty levels  $\lambda_1 \geq 0$ and $\lambda_2 \geq 0$: 
\begin{eqnarray}\label{eq2.2.1}
 \sh_{\lava}(z)&=&(1-k)z+ k(|z|-\lambda_1/(2k))_+\mathrm{sign}(z)\\
  &=&\left \{ \begin{array}{lll}
 z-\lambda_1/2,& z>\lambda_1/(2k)\\
(1-k)z, & -\lambda_1/(2k)\leq z \leq\lambda_1/(2k)\\
 z+ \lambda_1/2,& z<-\lambda_1/(2k)
 \end{array} \right. \label{eq2.2}
\end{eqnarray}
where $k:=\frac{\lambda_2}{1+\lambda_2}.$  The post-lava shrinkage function is given by
$$
\sh_{\mathrm{post}\text{-}\lava}(z)=\begin{cases}
 z,& |z|>\lambda_1/(2k),\\
(1-k)z, & |z|\leq\lambda_1/(2k).
 \end{cases}
$$
\end{lemma}

Figure \ref{f2} plots the lava shrinkage function along with  various alternative shrinkage functions for $z>0$. The top panel of the figure compares lava shrinkage to ridge, lasso, and elastic net shrinkage.  
It is clear from the figure that lava shrinkage is different from lasso, ridge, and elastic net shrinkage.  The figure also illustrates how lava provides a bridge between lasso and ridge, with the lava shrinkage function coinciding with the ridge shrinkage function for small values of the input $z$ and coinciding with the lasso shrinkage function for larger values of the input. 
%For comparison purposes, we use the same  tuning parameters $(\lambda_1,\lambda_2)$ respectively in  the $\ell_1$ and $\ell_2$ penalizations for all the methods. With these tuning parameters, 
Specifically, we see that the lava shrinkage function is a combination of lasso and ridge shrinkage that corresponds to using whichever of the lasso or ridge shrinkage is closer to the 45 degree line. 

It is also useful to consider how lava and post-lava compare with the post-lasso or hard-thresholding shrinkage: 
$\sh_{\text{post-lasso}}(z)=z1\{|z|>\lambda_l/2\}.$ These different shrinkage functions are illustrated in the bottom panel of Figure \ref{f2}.

\begin{figure}[htbp]
\begin{center}
\includegraphics[width=9cm]{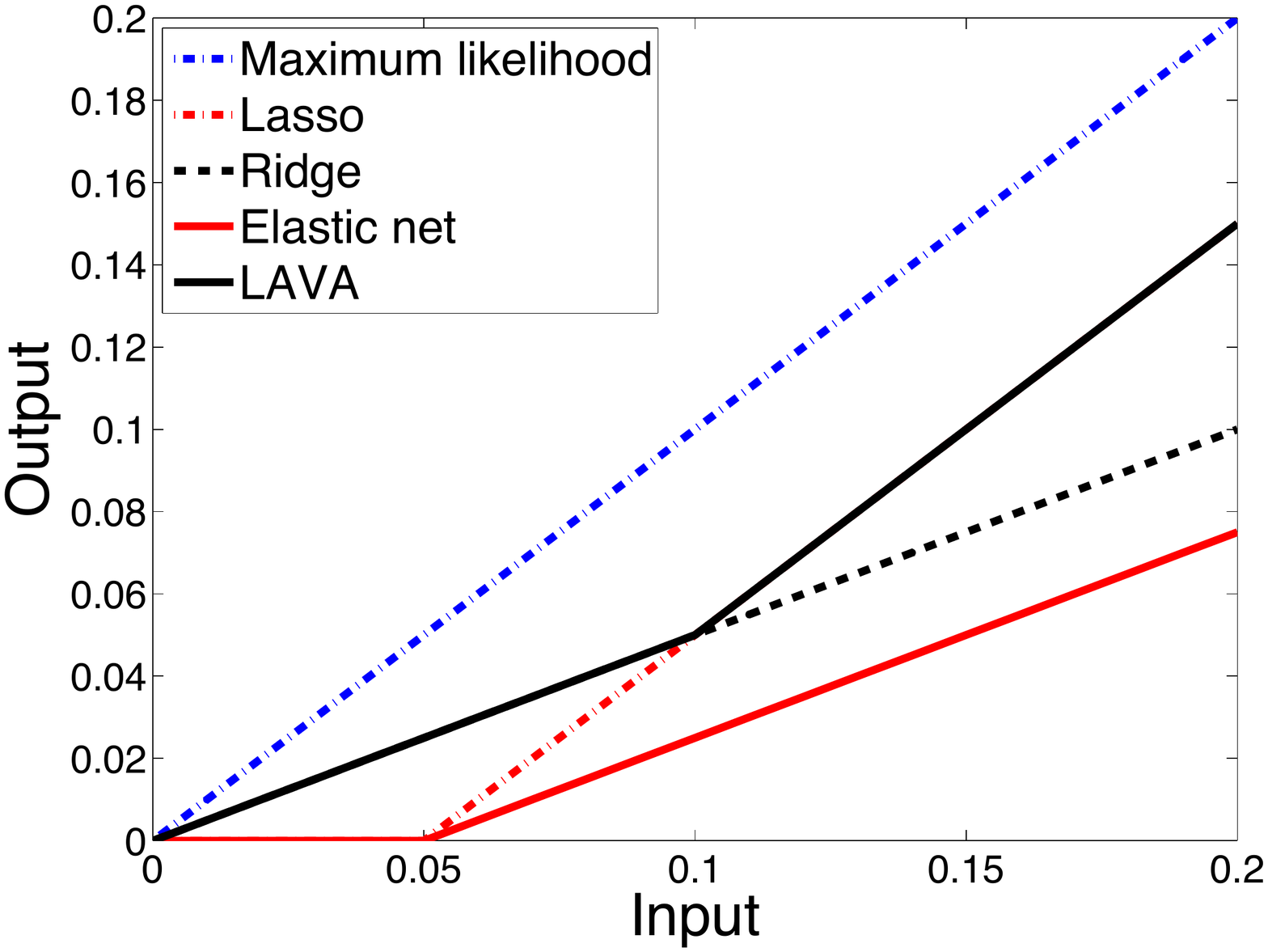}
\includegraphics[width=9cm]{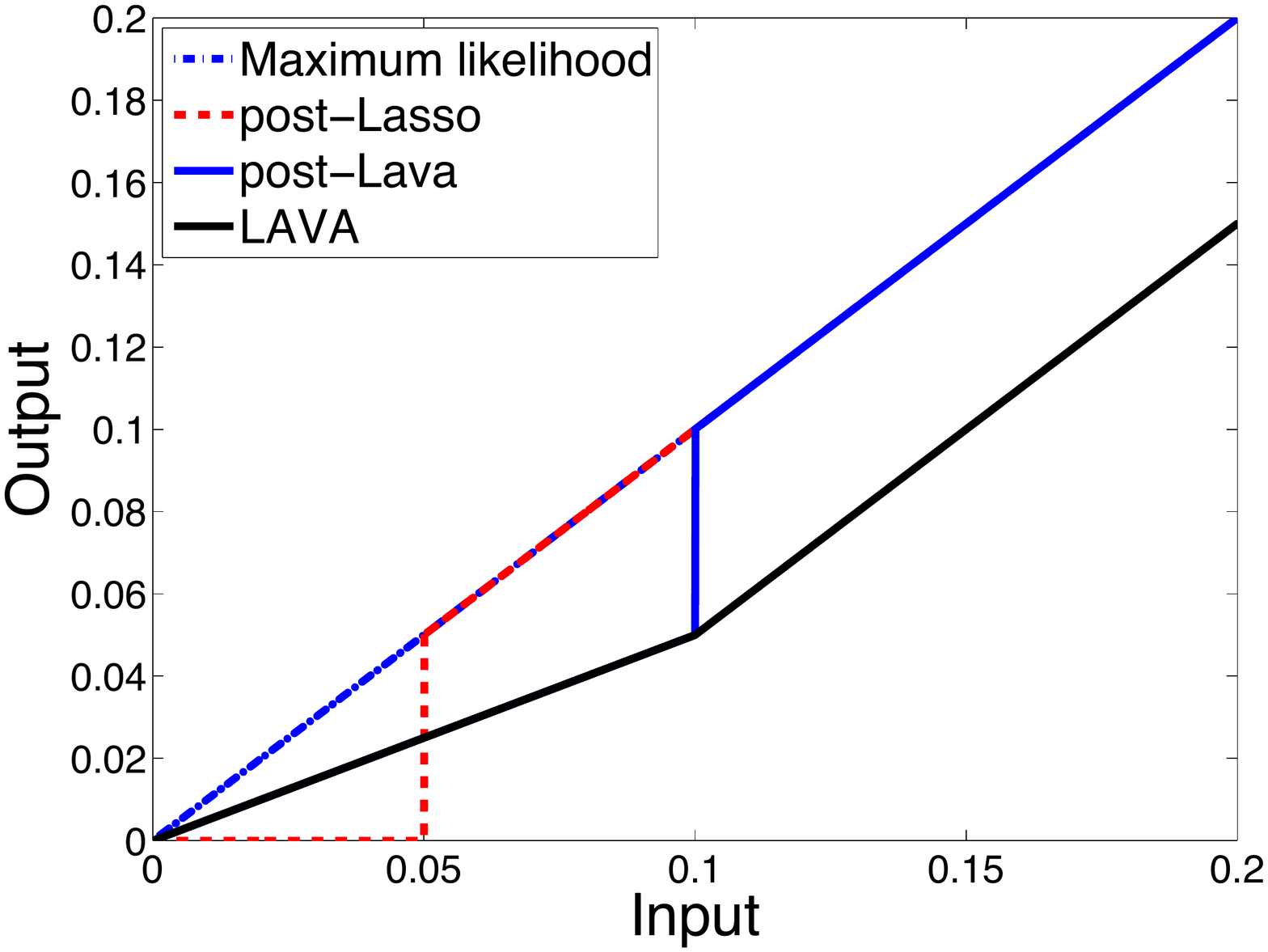}
\caption{\footnotesize Shrinkage functions.  Here we plot shrinkage functions implied by lava and various commonly used penalized estimators.  These shrinkage functions correspond to the case where penalty parameters are set as $\lambda_2 =\lambda_r = 1/2$  and $\lambda_1 = \lambda_l =1/2$.  In each figure, the light blue dashed line provides the 45 degree line coinciding to no shrinkage.  }
\label{f2}
\end{center}
\end{figure}

From (\ref{eq2.2.1}), we observe some key characteristics of the lava shrinkage function:

\textbf{1)} The lava shrinkage admits the lasso and ridge shrinkages as two extreme cases.
The lava and lasso shrinkage functions are the same when $\lambda_2=\infty$, and the ridge and lava shrinkage functions coincide if $\lambda_1 = \infty$.

\textbf{2)}  The lava shrinkage function $\sh_{\lava}(z)$ is a weighted average of data $z$ and the lasso shrinkage
function $\sh_{\lasso}(z)$ with weights given by  $1-k$ and $k$.

\textbf{3)}  The lava never produces a sparse solution when $\lambda_2<\infty$:  If $\lambda_2<\infty$, $\sh_{\lava}(z)=0$
if and only if $z=0$. This behavior is strongly different from elastic net which produces a sparse solution as long as $\lambda_1 > 0$.  %Of course, sparsity is not a requirement for making good predictions.

\textbf{4)} The lava shrinkage function continuously connects the ridge shrinkage function and the lasso shrinkage function.  When $|z|<\lambda_1/(2k)$, lava shrinkage is equal to ridge shrinkage; and when $|z|>\lambda_{1}/(2k)$, lava shrinkage is equal to lasso shrinkage.

\textbf{5)} The lava shrinkage does \textit{exactly the opposite} of the elastic net shrinkage.  The elastic net shrinkage function coincides with the lasso shrinkage function when $|z|<\lambda_1/(2k)$; and when $|z|>\lambda_{1}/(2k)$, the elastic net shrinkage is the same as ridge shrinkage.
%\end{itemize}

%The lava shrinkage has the same bias as lasso does when $|z|>\lambda_1/(2k)$. Such a bias can be removed through a \textit{post-lava} estimator. Specifically, the post-lava shrinkage is defined as:

\subsection{The risk function of the lava estimator in the one dimensional case}
In the one-dimensional case with $Z\sim N(\theta, \sigma^2)$, a natural measure of the risk of a given estimator $\widehat\theta = \sh(Z)$ is given by
\begin{eqnarray}\label{eq2.3}
 R(\theta, \widehat\theta)&=&\Ep[\sh(Z)-\theta]^2\cr
 &=&-\sigma^2+\Ep(Z-\sh(Z))^2+2\Ep[(Z-\theta)\sh(Z)].
\end{eqnarray}
 
 Let $\Pr_{\theta,\sigma}$ denote the probability law of $Z$.  Let $\phi_{\theta,\sigma}$ be the density function of $Z$.  We provide the risk functions of lava and post-lava in the following theorem.  We also present the risk functions of ridge, elastic  net, lasso, and post-lasso for comparison.     
\begin{theorem}[Risk Function of Lava and Related Estimators in the Scalar Case] \label{th2.1}
 Suppose $Z\sim N(\theta, \sigma^2)$.  Then for $w=\lambda_1/(2k)$, $k=\lambda_2/(1+\lambda_2)$, $h=1/(1+\lambda_2)$, $d=-\lambda_1/(2(1+\lambda_2))-\theta$ and $g=\lambda_1/(2(1+\lambda_2))-\theta$, we have 
   \begin{eqnarray*}
&& R(\theta,\widehat\theta_{\lava}) =  -k^2(w  +\theta)\phi_{\theta,\sigma}(w)\sigma^2+k^2( \theta-w)\phi_{\theta,\sigma}(-w)\sigma^2 \cr
&& \hspace{1in} +(\lambda_1^2/4+\sigma^2)\Pr_{\theta,\sigma}(|Z|>w) +(\theta^2k^2+(1-k)^2\sigma^2)\Pr_{\theta,\sigma}(|Z|<w),
\cr
&& R(\theta,\widehat\theta_{\mathrm{post}\text{-}\lava}) =  \sigma^2[-k^2w  +2kw -k^2\theta ]\phi_{\theta,\sigma}(w)+\sigma^2 [-k^2w+2kw+k^2\theta]\phi_{\theta,\sigma}(-w)\cr
&& \hspace{1in} + \sigma^2\Pr_{\theta,\sigma}(|Z|>w)+(k^2\theta^2+(1-k)^2\sigma^2)\Pr_{\theta,\sigma}(|Z|<w),\cr
&& R(\theta,\widehat\theta_{\lasso}) = -(\lambda_l/2  +\theta)\phi_{\theta,\sigma}(\lambda_l/2)\sigma^2+( \theta-\lambda_1/2)\phi_{\theta,\sigma}(-\lambda_l/2)\sigma^2 \cr
&&\hspace{1in} +(\lambda_l^2/4+\sigma^2)\Pr_{\theta,\sigma}(|Z|>\lambda_l/2) +\theta^2\Pr_{\theta,\sigma}(|Z|<\lambda_l/2),   \cr
&& R(\theta,\widehat\theta_{\mathrm{post}\text{-}\lasso}) =  (\lambda_l/2 -\theta )\phi_{\theta,\sigma}(\lambda_l/2)\sigma^2+ (\lambda_l/2+\theta)\phi_{\theta,\sigma}(-\lambda_l/2)\sigma^2\cr
&& \hspace{1in} + \sigma^2\Pr_{\theta,\sigma}(|Z|>\lambda_r/2)+\theta^2\Pr_{\theta,\sigma}(|Z|<\lambda_r/2),\cr
&& R(\theta,\widehat\theta_{\ridge}) =   \theta^2\widetilde k^2+(1-\widetilde k)^2\sigma^2,\quad \widetilde k=\lambda_r/(1+\lambda_r),\cr
&&R(\theta,\widehat\theta_{\mathrm{elastic  \ net}})= \sigma^2(h^2\lambda_1/2+h^2\theta+2dh) \phi_{\theta,\sigma}(\lambda_1/2)\cr
&& \hspace{1in}-\sigma^2(-h^2\lambda_1/2+h^2\theta+2gh)\phi_{\theta,\sigma}(-\lambda_1/2)+\theta^2\Pr_{\theta,\sigma}(|Z|<\lambda_1/2)\cr
&& \hspace{1in}+ ((h\theta+d)^2+h^2\sigma^2)\Pr_{\theta,\sigma}(Z>\lambda_1/2)\\
&& \hspace{1in} +((h\theta+g)^2+h^2\sigma^2)\Pr_{\theta,\sigma}(Z<-\lambda_1/2).
\end{eqnarray*}

\end{theorem}

%We shall rely on these results below to demonstrate how in the multidimensional cases the lava estimator can dominate or perform quite favorably when compared to the maximum likelihood estimator, the ridge estimator, and the lasso estimator.
 
These results for the one-dimensional case provide a key building block for results in the multidimensional case provided below.  In particular, we build from these results to show that the lava estimator performs very favorably relative to, and can substantially dominate, the maximum likelihood estimator, the ridge estimator,  and $\ell_1$-based   estimators (such as lasso and elastic-net) in interesting multidimensional settings.

\subsection{Multidimensional case}
 
We consider now the canonical Gaussian model or the Gaussian sequence
model.  In this case, we have that  
$$Z\sim N_p(\theta, \sigma^2I_p)$$ 
is a single observation from a multivariate normal distribution where $\theta=(\theta_1,...,\theta_p)'$ is a $p$-dimensional vector.   A fundamental result for this model is that the maximum likelihood estimator $Z$ is inadmissible and can be dominated by the ridge estimator and related shrinkage procedures when $p\geq 3$ (e.g. \cite{stein1956}).

In this model, the lava estimator is given by 
$$\widehat{\theta}_{\text{lava}} := (\widehat \theta_{\lava,1},..., \widehat \theta_{\lava, p})':=(\sh_{\text{lava}}(Z_1),...,\sh_{\text{lava}}(Z_p))',$$  
where $\sh_{\text{lava}}(z)$ is the lava shrinkage function as in (\ref{eq2.2}). %The soft-thresholding (lasso) shrinkage is particularly useful when $\theta$ is sparse, while the ridge shrinkage is applicable when $\theta$ is a dense vector with a small $l_2$-norm. 
%The estimator is designed to capture the superposition of sparse and dense signals, that is when $\theta$ can be decomposed into
The estimator is designed to capture the case where 
$$
\theta=\underbrace{\beta}_{\text{dense part}}+\underbrace{\delta}_{\text{sparse part}}
$$
is formed by combining a sparse vector $\delta$ that has a relatively small number of non-zero entries which are all large in magnitude and a dense vector $\beta$ that may contain very many small non-zero entries.  This model for $\theta$ is ``sparse+dense.'' It includes cases that are not well-approximated by ``sparse'' models - models in which a very small number of parameters are large and the rest are zero - or by ``dense'' models - models in which very many coefficients are non-zero but all coefficients are of similar magnitude.  This structure thus includes cases that pose challenges for estimators such as the lasso and elastic net that are designed for sparse models and for estimators such as ridge that are designed for dense models.

%Relative to the standard ``dense" and ``sparse" models, for which the ridge and the lasso methods are taylor-made, this model is ``dense-sparse" and cannot be well approximated by either a ``dense only" or ``sparse only" models.  The lava estimator is designed specifically for such general  models.

\begin{remark}
The regression model with Gaussian noise and an orthonormal design is a special case
of the multidimensional canonical model. Consider 
$$
Y=X\theta+U,\quad U \mid X \sim N(0,\sigma_u^2I_n),
$$
where $Y$ and $U$ are $n\times 1$ random vectors and $X$ is an
$n\times p$ random or fixed matrix, with $n$ and $p $ respectively denoting the sample size and the dimension of $\theta$. Suppose $\frac{1}{n}X'X=I_p$ a.s.. with $p \leq n$. Then we have the canonical multidimensional model:
$$
Z  =  \theta +  \epsilon,  \ \ Z = \frac{1}{n}X'Y, \quad \epsilon =  \frac{1}{n} X'U \sim N(0, \sigma^2 I_p ), \quad \sigma^2 = \frac{\sigma^2_u}{n}.  \quad \scriptstyle \blacksquare
$$

\end{remark}

All of the shrinkage estimators discussed in Section 2.1 generalize to the multidimensional case in the same way as lava.  Let $z \mapsto \sh_e(z)$ be the shrinkage function associated with estimator $e$ in the one dimensional setting where $e$ can take values in the set  $$ \mathcal{E} = \{ \lava, \mathrm{post\text{-}lava}, \ridge, \lasso, \mathrm{post\text{-}lasso}, \mathrm{elastic \ net} \}.$$  We then have a similar estimator in the multidimensional case given by
$$\widehat{\theta}_e:= (\widehat \theta_{e,1},..., \widehat \theta_{e, p})':=(\sh_{e }(Z_1),...,\sh_{ e}(Z_p))'.$$
The risk calculation from the one dimensional case then caries over to the multidimensional case since 
$$
\mathrm{R}(\theta, \widehat \theta_e) := \Ep \| \theta - \widehat \theta_e\|_2^2 =  \sum_{j=1}^p R(\theta_j, \widehat \theta_{e,j}).
$$
Given this fact, we immediately obtain the following result.

\begin{theorem}[Risk Function of Lava and Related Estimators in the Multi-Dimensional Case]\label{MultiDimLavaRisk} If $Z \sim N(0, \sigma^2 I_p)$, then for any $e \in \mathcal{E}$
we have that
$$
\mathrm{R}(\theta, \widehat \theta_e) = \sum_{j=1}^p R(\theta_j, \widehat \theta_{e,j}),
$$
where $R(\cdot, \cdot)$ is the uni-dimensional risk function 
characterized in Theorem \ref{th2.1}.
\end{theorem}

These risk functions are illustrated in Figure 1 in a prototypical ``sparse+dense" model generated according to the model discussed in detail in Section 2.5.  The tuning parameters used in this figure are the best possible (risk minimizing or oracle) choices of the  penalty levels found by minimizing the risk expression given in Theorem \ref{MultiDimLavaRisk}. 
 
\subsection{Canonical plug-in choice of penalty levels}
We now discuss simple, rule-of-thumb choices for the penalty levels for lasso ($\lambda_l$), ridge ($\lambda_r$) and lava ($\lambda_1, \lambda_2$). In the Gaussian model, a canonical choice of $\lambda_l$ is 
$$
\lambda_l=2\sigma{\Phi^{-1}(1-c/(2p))},
$$
which satisfies 
$$\Pr \left (\max_{j\leq p}|Z_j-\theta_j|\leq\lambda_l/2 \right)\geq 1-c;$$
see, e.g., \cite{donoho1995adapting}.
Here $\Phi(\cdot)$ denotes the standard normal cumulative distribution function, and $c$ is a pre-determined significance level which is often set to 0.05.  
The risk function for ridge is simple, and an analytic solution to the risk minimizing choice of ridge tuning parameter is given by
$$
\lambda_r=\sigma^2 (p/\|\theta\|_2^2).
$$

%ideal tuning parameter $\lambda_r$ for ridge can be obtained by minimizing  
%$$
%\sum_{j=1}^pR(\theta_j,\widehat\theta_{\ridge,j})=\widetilde k^2\|\theta\|_2^2+(1-\widetilde k)^2\sigma^2p,
%$$ 
%where $\widetilde k=\lambda_r/(1+\lambda_r)$, and is given by

As for the tuning parameters for lava, recall that the lava estimator in the Gaussian model is 
\begin{eqnarray*}
\widehat\theta_{\text{lava}}&=&(\widehat\theta_{\lava,1},...,\widehat\theta_{\lava,p})',\quad \widehat\theta_{\lava, j}=\widehat\beta_j+\widehat\delta_j,\quad j=1,...,p, \cr
(\widehat\beta_j,\widehat\delta_j)&=&\arg\min_{(\beta_j,\delta_j) \in \mathbb{R}^2}  (Z_j-\beta_j-\delta_j)^2+\lambda_2|\beta_j|^2+\lambda_1|\delta_j|.
\end{eqnarray*}
%Recall that lava is motivated by a model where $\theta=\beta+\delta$. 
If the dense component $\beta$ were known, then following \cite{donoho1995adapting} would suggest setting 
$$
\lambda_1=2\sigma{\Phi^{-1}(1-c/(2p))}
$$
as a canonical choice of $\lambda_1$ for estimating $\delta$.
If the sparse component $\delta$ were known, one could adopt 
$$
\lambda_2=\sigma^2(p/\|\beta\|_2^2)
$$
as a choice of $\lambda_2$ for estimating $\beta$ following the logic for the standard ridge estimator.

We refer to these choices as the ``canonical plug-in" tuning parameters  and use them in constructing the risk comparisons in the following subsection.  We note that the lasso choice is motivated by a sparse model and does not naturally adapt to or make use of the true structure of $\theta$.  The ridge penalty choice is explicitly tied to risk minimization and relies on using knowledge of the true $\theta$.  The lava choices for the parameters on the $\ell_1$ and $\ell_2$ penalties are, as noted immediately above, motivated by the respective choices in lasso and ridge.  As such, the motivations and feasibility of these canonical choices are not identical across methods, and the risk comparisons in the following subsection should be interpreted within this light. 
 
\subsection{Some risk comparisons in a canonical Gaussian model}
 
To compare the risk functions of lava, lasso, and ridge
estimators, we consider a canonical Gaussian model, where 
$$\theta_1=3 ,  \quad \theta_j=0.1q,\quad j=2,...,p,$$
 for some $q\geq 0$.  We set the noise level to be  $\sigma^2=0.1^2$.
  The parameter $\theta$ can be decomposed as 
 $\theta=\beta+\delta$, where 
the sparse component  is $\delta=(3,0,...,0)'$, 
and the dense component  is $$\beta=(0, 0.1q,...,0.1q)',$$
where $q$ describes the ``size of small coefficients."  The 
canonical tuning parameters are 
$\lambda_l=\lambda_1=2\sigma{\Phi^{-1}(1-c/(2p))}$, $\lambda_r=\sigma^2p/(3+0.1^2q^2(p-1))$ and $\lambda_2=\sigma^2p/(0.1^2q^2(p-1))$.

 Figure \ref{f1} (given in the introduction) compares risks of lava, lasso, ridge, elastic net,  
and the maximum likelihood estimators as functions of 
the size of the small coefficients  $q$, using the ideal (risk minimizing or oracle choices) of the penalty levels. Figure \ref{f3}  compares risks of lava, lasso, ridge    and the maximum likelihood estimators using the ``canonical plug-in" penalty levels discussed above.  Theoretical risks are plotted as a function of  the size of the small coefficients  $q$.  We see from these figures that regardless of how we choose the penalty levels -- ideally or via the plug-in rules -- lava strictly dominates the  competing methods in this ``sparse+dense" model. Compared to lasso, the proposed lava estimator does about as well as lasso when the signal is sparse and does significantly better than lasso when the signal is non-sparse. Compared to ridge, the lava estimator does about as well as ridge when the signal is dense and does significantly better than ridge when the signal is sparse. %Both the ridge and lava estimators dominate the maximum likelihood estimator across all values of $q$

\begin{figure}[htbp]
\begin{center}
\includegraphics[width=9cm]{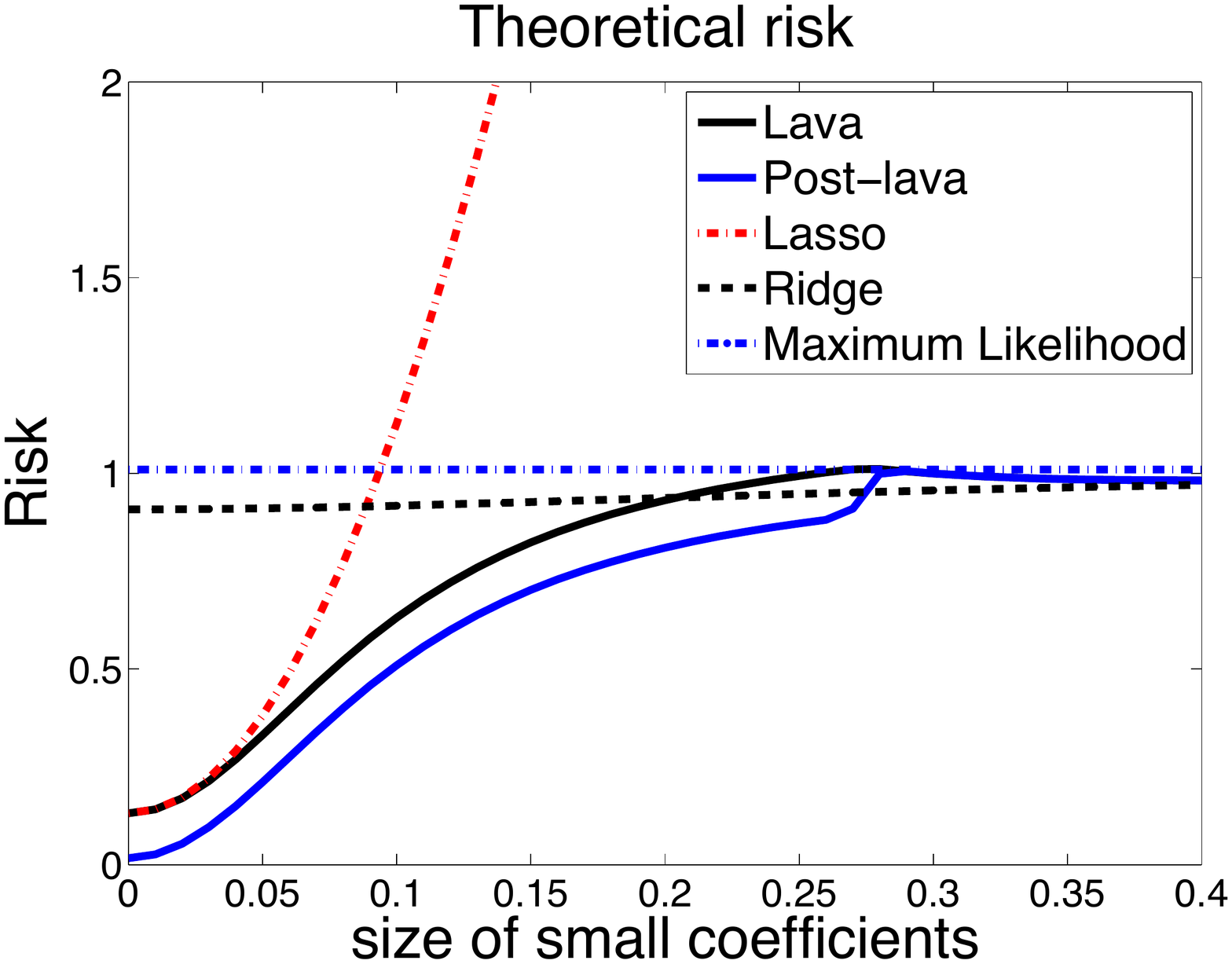}
\includegraphics[width=9cm]{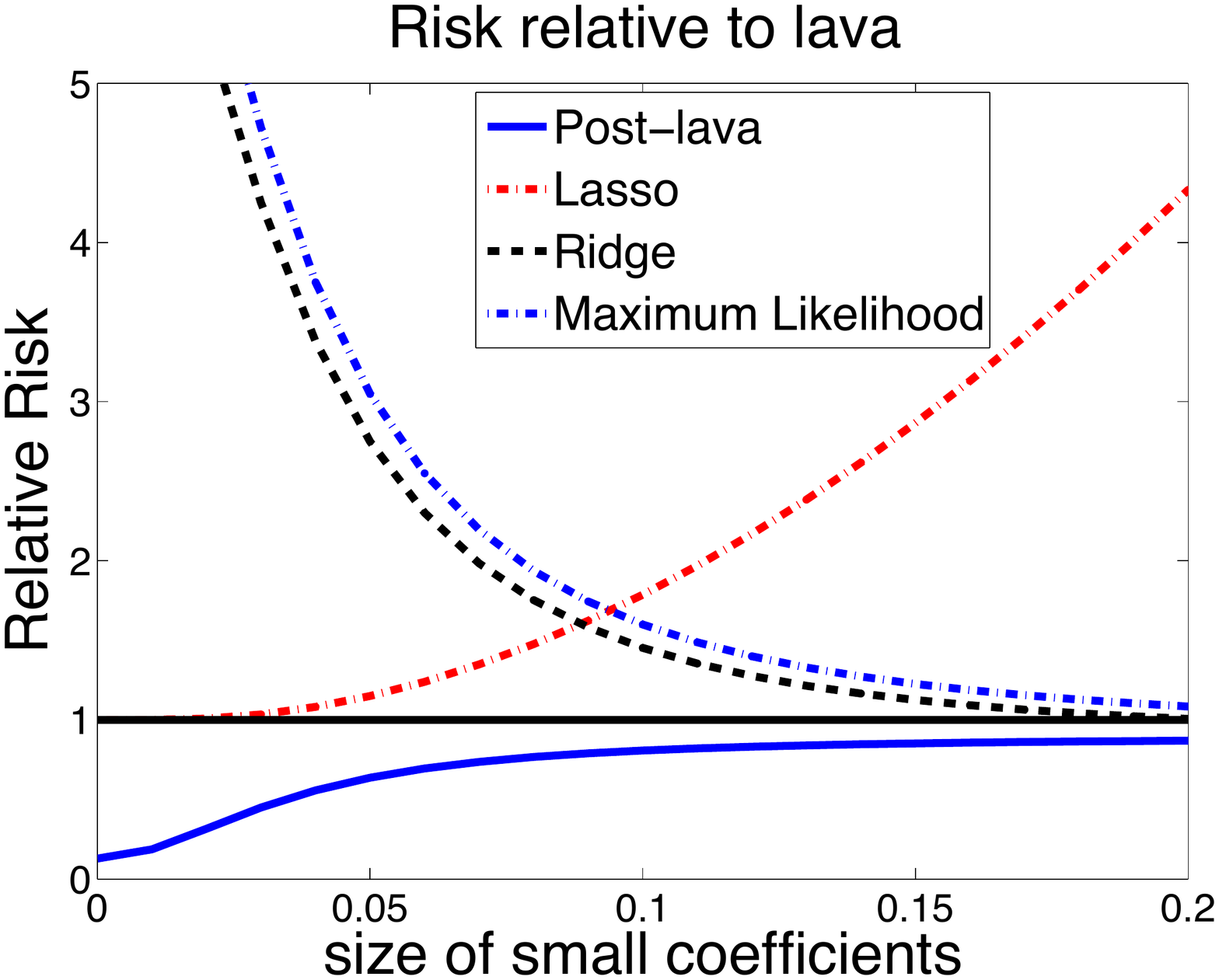}
\caption{\footnotesize Exact risk functions of lava, post-lava, ridge, lasso,   and maximum likelihood in the Gaussian sequence model with ``sparse+dense'' signal structure, using the canonical ``plug-in" choices of penalty levels. See Section 2.5 for the description of penalty levels and the model. The size of ``small coefficients" is shown on the horizontal axis. The size of these coefficients directly corresponds to the size of the ``dense part" of  the signal, with zero corresponding to the exactly sparse case.   Relative risk plots the ratio of the risk of each estimator to the lava risk, $\mathrm{R}(\theta, \widehat \theta_e)/\mathrm{R}(\theta, \widehat \theta_{\lava})$.  Note that the relative risk plot is over a smaller set of sizes to accentuate comparisons over the region where there are the most interesting differences between the estimators.
}
\label{f3}
\end{center}
\end{figure}

In Section 5 we further explore the use of feasible, data-driven choices of penalty levels via cross-validation and SURE minimization; see Figures \ref{f4} and 5.  We do so in the context of the Gaussian regression model with fixed regressors.   With  either cross-validation or SURE minmization, the ranking of the estimators remains unchanged, with lava consistently dominating lasso, ridge, and the elastic net.

 \cite{stein1956} proved that a ridge estimator strictly dominates maximum likelihood in the Gaussian sequence model once $p \geq 3$.  In the comparisons above, we also see that the lava estimator strictly dominates the maximum likelihood estimator; and one wonders whether this domination has a theoretical underpinning similar to Stein's result for ridge. The following result provides some (partial) support for this phenomenon for the lava estimator with the plug-in penalty levels.  The result shows that, for a sufficiently large $n$ and $p$, lava does indeed uniformly dominate the maximum likelihood estimator on the compact set $\{\theta=\beta+\delta: \|\beta\|_{\infty}+\|\delta\|_{\infty}<M\}$. 
%Specifically, we have the following result. 
 
\begin{lemma}[Relative Risk of Lava vs. Maximum Likelihood ]\label{th2.2}
Suppose $Z\sim N_p(\theta,\sigma^2I_p)$, where $\theta$ can be decomposed into $\theta=\beta+\delta$ with $s=\sum_{j=1}^p1\{\delta_j\neq0\}<p$.  
Let $\lambda_1$ and $\lambda_2$ be chosen with the plug-in rule given in Section 2.4.  Then uniformly for $\theta\in \{\theta=\beta+\delta: \|\beta\|_{\infty}+\|\delta\|_{\infty}<M\}$, 
when  $\sigma\sqrt{\log p}>2M+33\sigma$,   $ M^2\log p>16\sigma^2$,  and $    \pi c^2 {\log p}\geq  1$,  we have
$$
RR := \frac{ \Ep\| \widehat\theta_{\lava}(Z)-\theta\|_2^2}{\Ep\|Z-\theta\|_2^2}\leq \frac{\|\beta\|_2^2}{\sigma^2p+\|\beta\|_2^2}+\frac{3sM^2}{p\sigma^2}+ \frac{4}{\sqrt{2\pi}p^{1/16}}\left( 1+\frac{7M}{\sigma p^{1/16}} \right).$$

\end{lemma}

\begin{remark} Note that
$$
R^2_d = \frac{\|\beta\|_2^2}{\sigma^2p+\|\beta\|_2^2}
$$
measures the proportion of the total variation of $Z-\delta$ around $0$ that is explained by the dense part of the signal. 
If $R^2_d$ is bounded away from $1$ and $M$ and $\sigma^2>0$ are fixed, 
then the risk of lava becomes uniformly smaller than the risk of the maximum likelihood estimator on  a compact parameter space as $p\to\infty$ and $s/p\to0$.  Indeed, if $R^2_d$ is bounded away from $1$, then
$$
\frac{3sM^2}{p\sigma^2}+ \frac{4}{\sqrt{2\pi}p^{1/16}}\left( 1+\frac{7M}{\sigma p^{1/16}} \right) \to 0
\implies RR = R^2_d + o(1) <  1.
$$
%$$ \frac{3sM^2}{p\sigma^2}+\frac{1}{\sqrt{2\pi}}\frac{2}{p^{1/64}} \left(    \frac{18M}{\sigma}+ 8 \right) \to 0\implies \mathrm{RR} = R^2 + o(1) <  1.$$
Moreover, we have $RR \to 0$ if $R^2_d \to 0 $.  That is, the lava estimator becomes infinitely more asymptotically efficient than the maximum likelihood estimator in terms of relative risk.
% Note that the upper bound is  monotone  in $\|\beta\|_2^2$. Hence  the gain of using lava is more substantial for small values of $\|\beta\|_2^2$, which is also illustrated in Figure \ref{f2}. 
 \qed
  
\end{remark}

% \begin{remark}
%The above theorem does not require the identification of the decomposition $\theta=\beta+\delta$. The result holds uniformly for %all pairs $(\beta,\delta)$ that satisfy the theorem's conditions.
 % \end{remark}

\subsection{Stein's unbiased risk estimation for lava}
 
\cite{stein1981estimation} proposed a useful risk estimate based on the integration by parts formula, now commonly referred to as \textit{Stein's unbiased risk estimate} (SURE).    This subsection derives SURE for the lava shrinkage in the multivariate  Gaussian model.  
 
 Note that
\begin{equation}\label{eq2.5} \Ep\|\widehat\theta_{\text{lava}}
-\theta\|_2^2=-p\sigma^2+\Ep \|Z-\widehat\theta_{\text{lava}}\|_2^2
+2\Ep[(Z-\theta)'\widehat\theta_{\text{lava}}].
 \end{equation}

An essential component to understanding the risk is given by applying Stein's formula to calculate $ \Ep[(Z-\theta)'\widehat\theta_{\text{lava}}]$.  A closed-form expression for this expression in the one-dimensional case is given in equation (\ref{eq2.4}) in the appendix.
%Recall that  $\widehat\theta_{\text{lava}}(z)=(1-k)z+k\widehat\delta(z)$ is  a weighted average of the data and soft thresholding with shrinkage parameter $\lambda_1/(2k)$, $k=\lambda_2/(1+\lambda_2)$.
 The following result provides the SURE for lava in the more general multidimensional case.
  \begin{theorem}[SURE for lava]\label{th2.3} Suppose $Z=(Z_1,...,Z_p)'
  \sim N_p(\theta,\sigma^2I_p)$. Then
$$
\Ep[(Z-\theta)'\widehat\theta_{\lava}]=p(1-k)\sigma^2 +k\sigma^2\sum_{j=1}^p\Pr_{\theta_j,\sigma}(|Z_j|>\lambda_1/(2k)).
$$
 In addition,  let $\{Z_{ij}\}_{i=1}^n$ be identically distributed as $Z_j$
 for each $j$. Then
   $$
\widehat{\mathrm{R}}(\theta, \widehat \theta_{\lava})= ( 1-2k)p\sigma^2+
\frac{1}{n} \sum_{i=1}^n \|Z_i-\sh_{\lava}(Z_i)\|_2^2 +2k\sigma^2\frac{1}{n}\sum_{i=1}^n\sum_{j=1}^p1\{|Z_{ij}|>\lambda_1/(2k)\}.
 $$
  is an unbiased estimator of $\mathrm{R}(\theta, \widehat \theta_{\lava})$. 
 
  \end{theorem}

 \section{Lava in the Regression Model}
 
\subsection{Definition of Lava in the Regression Model}
Consider a fixed design regression model:
$$
Y=X\theta_0+U,\quad U\sim N(0,\sigma_u^2I_n),
$$
where $Y=(y_1,...,y_n)'$, $X=(X_1,...,X_n)'$, and 
$\theta_0$ is the true regression coefficient.  Following the previous discussion, we assume that $$\theta_0=\beta_0+\delta_0$$ is ``sparse+dense'' with sparse component $\delta_0$ and dense component $\beta_0$.  Again, this coefficient structure includes cases which cannot be well-approximated by traditional sparse models or traditional dense models and will pose challenges for estimation strategies tailored to sparse settings, such as lasso and similar methods, or strategies tailored to dense settings, such as ridge. 

In order to define the estimator we shall rely on the normalization condition that
\begin{equation}\label{normalization}
 n^{-1} [X'X]_{jj} =1, \quad j=1,...,p.
\end{equation}
Note that without this normalization, the penalty terms below would have to be modified
in order to insure equivariance of the estimator to changes of scale in the columns of $X$.

The lava estimator $\widehat \theta_{\lava}$ of $\theta_0$  solves the following optimization problem:
 \begin{eqnarray}
 \widehat\theta_{\text{lava}}&:=&\widehat\beta+\widehat\delta,\cr
 (\widehat\beta,\widehat\delta)&:=&\arg\min_{(\beta',\delta')'\in\mathbb{R}^{2p} } \left \{ \frac{1}{n}\|Y-X(\beta+\delta)\|_2^2+\lambda_2\|\beta\|_2^2+\lambda_1\|\delta\|_1 \right \}.
 \end{eqnarray}
The lava program splits parameter $\theta$ into the sum of $\beta$ and $\delta$ and penalizes
these two parts using the $\ell_2$ and $\ell_1$ penalties.   Thus, the $\ell_1$- penalization regularizes
the estimator of the sparse part $\delta_0$ of $\theta_0$ and produces a sparse solution $\widehat\delta$. The $\ell_2$-penalization regularizes the estimator of the dense part $\beta_0$ of $\theta_0$ and produces a dense solution $\widehat \beta$. The resulting estimator of $\theta_0$ is then simply the sum of the sparse estimator $\widehat\delta$ and the dense estimator $\widehat\beta$. 

\subsection{A Key Profile Characterization and Some Insights.}
The lava estimator can be computed in the following way. For a  fixed $\delta$, we minimize 
$$
\widehat\beta(\delta)=\arg\min_{\beta\in\mathbb{R}^p} \left \{ \frac{1}{n}\|Y-X(\beta+\delta)\|_2^2+\lambda_2\|\beta\|_2^2 \right\},
$$
with respect to  $\beta$.  This program is simply the well-known ridge regression problem, and the solution is
$$
\widehat\beta(\delta)=(X'X+n\lambda_2I_p)^{-1}X'(Y-X\delta).
$$
By substituting $\beta=\widehat\beta(\delta)$ into the objective function, we then define an $\ell_1$-penalized quadratic program %``profiled" optimization program 
which we can solve for $\widehat\delta$:
\begin{equation}\label{lava.p}
\widehat\delta=\arg\min_{\delta\in\mathbb{R}^p} \left \{\frac{1}{n}\|Y-X(\widehat\beta(\delta)+\delta)\|_2^2+\lambda_2\|\widehat\beta(\delta)\|_2^2+\lambda_1\|\delta\|_1 \right\}.
\end{equation}
%which is an $\ell_1$-penalized quadratic problem. 
The lava solution is then given by $\widehat\theta=\widehat\beta(\widehat\delta)+\widehat\delta$. The following result provides a useful %(and surprising, at least to us) 
characterization of the solution.

\begin{theorem}[A Key Characterization of the Profiled Lava Program]
 \label{th3.1}  Define ridge-projection matrices,
$$
\PR_{\lambda_2}=X(X'X+n\lambda_2 I_p)^{-1}X' \ and \ 
\K_{\lambda_2}=I_n-\PR_{\lambda_2},
$$ 
and transformed data,
$$
\widetilde Y=\K_{\lambda_2}^{1/2}Y \ and \ \widetilde X=\K_{\lambda_2}^{1/2}X.
$$
Then
\begin{equation}\label{eq3.2}
\widehat\delta=\arg\min_{\delta \in \mathbb{R}^p} \left \{ \frac{1}{n}\|\widetilde Y-\widetilde X\delta\|_2^2+\lambda_1\|\delta\|_1 \right\}
\end{equation}
and 
\begin{equation}\label{eq3.3}
X\widehat\theta_{\lava}=\PR_{\lambda_2}Y+\K_{\lambda_2}X\widehat\delta.
\end{equation}
\end{theorem}

The theorem shows that solving for the sparse part $\widehat\delta$ of the lava estimator is equivalent to solving for the parameter in a standard lasso problem using  transformed data. This result is key to both computation and our theoretical analysis of the estimator.

\begin{remark}[Insights derived from Theorem \ref{th3.1}] 
Suppose $\delta_0$ were known.  Let $W=Y-X\delta_0$ be the response vector after removing the sparse signal, and note that we equivalently have $W=X\beta_0+U$.  A natural estimator for $\beta_0$ in this setting is then the ridge estimator of $W$ on $X$:
$$
\widehat\beta(\delta_0)=(X'X+n\lambda_2I_p)^{-1}X'W.
$$
Denote the prediction error based on this ridge estimator as
$$
\textsf{D}_{\text{ridge}}(\lambda_2)=X\widehat\beta(\delta_0)-X\beta_0=  -\K_{\lambda_2}X\beta_0+\PR_{\lambda_2}U.
$$
Under mild regularity conditions on $\beta_0$ and the design matrix, \cite{hsu2014random} showed that 
$$
\frac{1}{n}\|\textsf{D}_{\text{ridge}}(\lambda_2)  \|^2 = o_P(1).$$
%at some rate.

Using Theorem \ref{th3.1}, the prediction error of lava can be written as
\begin{equation}\label{eq3.4add}
X\widehat\theta_{\text{lava}}-X\theta_0=\PR_{\lambda_2}Y+\K_{\lambda_2}X\widehat\delta-X\beta_0-X\delta_0=\textsf{D}_{\text{ridge}}(\lambda_2)+\K_{\lambda_2}X(\widehat\delta-\delta_0).
\end{equation}
Hence, lava has vanishing prediction error as long as 
\begin{equation}\label{lasso.part.consistent} 
\frac{1}{n}\|\K_{\lambda_2}X(\widehat\delta-\delta_0)\|_2^2=o_P(1).
\end{equation} 

Condition (\ref{lasso.part.consistent}) is related to the performance 
of the lasso in the transformed problem (\ref{eq3.2}). Examination of (\ref{eq3.2}) shows
that it corresponds to a sparse regression model with approximation errors $\K_{\lambda_2}^{1/2} X\beta_0$: 
For $\widetilde U= \K_{\lambda_2}^{1/2} U$,
\begin{equation}\label{decompose}
\widetilde Y= \widetilde X\delta_0+ \widetilde U+ \K_{\lambda_2}^{1/2} X\beta_0.
\end{equation}
Under conditions such as those given in \cite{hsu2014random}, the approximation error obeys
\begin{equation}\label{approx.error.small}
\frac{1}{n}\|\K_{\lambda_2}^{1/2} X\beta_0\|_2^2 = o_P(1).
\end{equation}
It is also known that the lasso estimator performs well in sparse models with vanishing approximation errors.  The lasso estimator attains rates of convergence in the prediction norm that are the sum of the usual rate of convergence in the case without approximation errors and the rate at which the approximation error vanishes; see, e.g., \cite{belloni2013least}. Thus, we anticipate that (\ref{lasso.part.consistent}) will hold.
%Therefore, we anticipate that the lasso on the transformed data $(\widetilde Y, \widetilde X)$ should consistently estimate $\delta_0$, so that (\ref{lasso.part.consistent}) holds. 

To help understand the plausibility of condition (\ref{approx.error.small}), consider an orthogonal design where $\frac{1}{n}X'X=I_p$. In this case, it is straightforward to verify that $\K_{\lambda_2}^{1/2}
=\K_{\lambda_2^*}$ where $\lambda_2^*= \sqrt{\lambda_2}/(\sqrt{1+\lambda_2}-\sqrt{\lambda_2}).$
Hence, 
$\widetilde X\beta_0=\K_{\lambda_2^*}X\beta_0$ is a component of the prediction bias from a ridge estimator with tuning parameter $\lambda_2^*$ and is stochastically negligible.  We present the rigorous asymptotic analysis for the general case in Section 3.5.  \qed
\end{remark}

\subsection{Degrees of Freedom and SURE}

Degrees of freedom is often used to quantify model complexity and to construct adaptive model selection criteria for selecting tuning parameters. In a  Gaussian linear regression model $Y\sim N(X\theta_0, \sigma_u^2I_n)$ with a fixed design, we can define the degrees of freedom of the mean fit $X\widehat\theta$ to be
$$
\df(\widehat\theta)=\frac{1}{\sigma_u^2}\Ep[(Y-X\theta_0)'X\widehat\theta];
$$
see, e.g., \cite{efron2004estimation}.
Note that this quantity is also an important component of the mean squared prediction risk:
$$
\Ep\frac{1}{n}\|X\widehat\theta-X\theta_0\|_2^2=-\sigma_u^2+\Ep\frac{1}{n}\|X\widehat\theta-Y\|_2^2+ \frac{2\sigma_u^2}{n} \df(\widehat\theta).
$$
 
\cite{stein1981estimation}'s SURE theory provides a tractable way of deriving an unbiased estimator of the degrees of freedom, and thus the mean squared prediction risk. Specifically,  write $\widehat\theta=d(Y,X)$ as a function of $Y$, conditional on $X$. Suppose  $d(\cdot, X):\mathbb{R}^n\to\mathbb{R}^p$ is  almost differentiable; see \cite{meyer2000degrees} and \cite{efron2004least}). For
 $f: \mathbb{R}^n  \to \mathbb{R}^n$ differentiable at $y$, define
$$
\partial_y  f(y) :=  [  \partial f_{ij} (y) ], \quad (i, j) \in\{1,...,n\}^2,  \quad 
\partial f_{ij} (y) : = \frac{\partial}{\partial y_j} f_i(y),
$$ $$
 \nabla_y\cdot f(y):=\text{tr} (\partial_y f(y)).
$$
Let $X_i'$ denote the $i$-th row of $X$, $i=1,...,n$.  Then, from \cite{stein1981estimation}, we have that
 $$
 \frac{1}{\sigma_u^2}\Ep[(Y-X\theta_0)'Xd(Y,X)]=\Ep[ \nabla_y\cdot (Xd(Y,X))] = \text{tr} \big (\partial_y[Xd(Y,X)] \big ). $$
An unbiased estimator of the term on the right-hand-side of the display may then be constructed using its sample analog.
%The quantity on the right side of the display can be estimated without a bias by the sample analogue.

In this subsection,  we derive the degrees of freedom of the lava, and thus a SURE of its mean squared prediction risk.  
By Theorem \ref{th3.1}, 
\begin{eqnarray}\label{eq3.4}
\nabla_y \cdot (Xd_{\text{lava}}(y,X)) & = & \text{tr}(\PR_{\lambda_2})+\nabla_y \cdot(\K_{\lambda_2} X d_{\lasso}(\K_{\lambda_2}^{1/2}y, \K_{\lambda_2}^{1/2}X)) \\
& = & \text{tr}(\PR_{\lambda_2}) + \text{tr}\left( \K_{\lambda_2}  \partial_y  [X d_{\lasso}(\K_{\lambda_2}^{1/2}y, \widetilde X)] \right),
\end{eqnarray}
where $d_{\lava}(y,X)$ is the lava estimator on the data $(y,X)$ and $d_{\lasso}(\K_{\lambda_2}^{1/2}y, \K_{\lambda_2}^{1/2}X)) $  is the lasso estimator on the data $ (\K_{\lambda_2}^{1/2}y, \K_{\lambda_2}^{1/2}X)$  with the penalty level $\lambda_1$. The almost differentiability of the map  $y \mapsto d_{\lasso}(\K_{\lambda_2}^{1/2}y, \K_{\lambda_2}^{1/2}X)$  follows from the almost differentiability of the map $u \mapsto d_{\lasso}(u, \K_{\lambda_2}^{1/2}X)$, which holds by
the results in \cite{dossal2011degrees} and \cite{tibshirani2012degrees}. 

The following theorem presents the degrees of freedom and SURE for lava.  Let $\widehat J=\{j\leq p: \widehat\delta_j\neq 0\}$  be the active set of the sparse component estimator with cardinality denoted by  $|\widehat J|$. Recall that $\widetilde X=\K_{\lambda_2}^{1/2}X$. Let $\widetilde X_{\hat J}$  be an $n\times |\widehat J|$ submatrix of $\widetilde X$ whose columns are those corresponding to the entries in $\widehat J$.  Let $A^-$ denote the Moore-Penrose pseudo-inverse of a square matrix $A$.

\begin{theorem}[SURE for Lava in Regression]\label{th3.2} Suppose $Y\sim N(X\theta_0,\sigma_u^2I_n)$.  Let 
$$
 \widetilde \K_{\hat J}=I-\widetilde X_{\hat J} ( \widetilde X_{\hat J} ' \widetilde X_{\hat J} )^- \widetilde X_{\hat J} '
$$
be the projection matrix onto the unselected columns  of the transformed variables. 
We have that $$
\df(\widehat\theta_{\lava})=\Ep[\rank(\widetilde X_{\hat J})+\tr(\widetilde \K_{\hat J}\PR_{\lambda_2})].
$$
Therefore, the SURE of $\Ep \frac{1}{n}\|X\widehat\theta_{\lava}-X\theta_0\|_2^2$ is given by
$$
-\sigma_u^2+\frac{1}{n}\|X\widehat\theta_{\lava}-Y\|_2^2+ \frac{2\sigma_u^2}{n} \rank(\widetilde X_{\hat J})+\frac{2\sigma_u^2}{n}\tr(\widetilde \K_{\hat J}\PR_{\lambda_2}).
$$
 \end{theorem}

\subsection{Post-lava in regression}
 
We can also remove the shrinkage bias in the sparse component introduced by the $\ell_1$-penalization via a \textit{post-selection} procedure. Specifically,  let ($\widehat\beta,\widehat\delta$) respectively denote the lava estimator of the dense and sparse components.
Define the \textit{post-lava} estimator  as follows:
 \begin{eqnarray*}
 \widehat \theta_{\text{post-lava}}&=&\widehat\beta+\widetilde\delta,\cr
 \widetilde\delta&=& \arg\min_{\delta \in \mathbb{R}^p} \left \{ \frac{1}{n}\|Y-X\widehat\beta-X\delta\|_2^2:\quad \delta_j=0\text{ if } \widehat\delta_j=0 \right \}.
 \end{eqnarray*}
Let $X_{\hat J}$  be an $n\times |\widehat J|$ submatrix of $X$ whose columns are selected by $\widehat J$.   Then we can partition $\widetilde\delta=(\widetilde\delta_{\hat J}, 0)'$,   where
 $\widetilde \delta_{\hat J}= ( X_{\hat J} '   X_{\hat J} )^-  X_{\hat J} '(Y-X\widehat\beta)$.   Write $\PR_{\hat J}=X_{\hat J}(X_{\hat J}'X_{\hat J})^-X_{\hat J}'$ and $\K_{\hat J}=I_n-\PR_{\hat J}$.
The post-lava prediction for $X\theta$ is:
$$
 X\widehat \theta_{\text{post-lava}}= \PR_{\hat J}Y+\K_{\hat J}X\widehat\beta.
$$
 In addition, note that  the lava estimator satisfies $X\widehat\beta= \PR_{\lambda_2}(Y- X \widehat\delta)$.  We then have the following expression of $ X\widehat \theta_{\text{post-lava}}$. \begin{lemma} \label{l3.1} Let $\widehat U:= Y-X\widehat\theta_{\lava}$. Then
 $
  X\widehat \theta_{\mathrm{post}\text{-}\lava}= X\widehat\theta_{\lava}+\PR_{\hat J}\widehat U.
 $
 \end{lemma}
 The above lemma  reveals that the post-lava corrects the $\ell_1$-shrinkage bias of the original lava fit by adding the projection of the  lava residual onto the subspace of the  selected regressors. This correction is in the same spirit as the post-lasso correction for shrinkage bias in the standard lasso problem; see \cite{belloni2013least}.

 \begin{remark}
We note that  the  SURE for post-lava may not exist, though an estimate of the upper bound of the risk function may be available, because of the impossibility results for constructing unbiased estimators for non-differentiable functions; see \cite{hirano2012impossibility}. \qed
 
 \end{remark}
 
%\begin{comment}
%The following theorem gives the degrees of freedom  of post-lava.
%
%\begin{theorem}\label{df_post_lava} The degrees of freedom of post-lava is 
%\begin{eqnarray*}
%\df(\widehat\theta_{\text{post-lava}})&=&\Ep\tr(I-K_{\hat J}\K_{\lambda_2}^{1/2}\widetilde K_{\hat J}\K_{\lambda_2}^{1/2})\cr
%&=&\df(\widehat\theta_{\text{lava}})+\Ep\tr(\Pr_{\hat J}\K_{\lambda_2}^{1/2}\widetilde K_{\hat J}\K_{\lambda_2}^{1/2}).
%\end{eqnarray*}
%So  the SURE of $E\frac{1}{n}\|X\widehat\theta_{\text{post-lava}}-X\theta_0\|_2^2$ is given by
%$$
%-\sigma_u^2+\frac{1}{n}\|X\widehat\theta_{\text{post-lava}}-Y\|_2^2+ \frac{2\sigma_u^2}{n}\tr(I-K_{\hat J}\K_{\lambda_2}^{1/2}\widetilde K_{\hat J}\K_{\lambda_2}^{1/2}).
%$$
%%In addition, 
%
%\end{theorem} 
%\begin{remark}
%By  letting $\lambda_2=\infty$, it follows that the degrees of freedom of post-lasso is:
%$$
%\df(\widehat\theta_{\text{post-lava}})=\Ep\rank(X_{\hat J}),
%$$
%which is the same as that of Lasso.
%\end{remark}
%
%\end{comment}

\subsection{Deviation Bounds for Prediction  Errors}

In the following, we develop deviation bounds for the lava prediction error: $\frac{1}{n}\|X\widehat\theta_{\text{lava}}-X\theta_0\|_2^2$. We continue to work with the decomposition $\theta_0=\beta_0+\delta_0$ and will show that lava performs well in terms of rates on the prediction error in this setting. 
%where the decomposition will be determined at the very end to yield the best performance bound for the lava estimator.  
According to the discussion in Section 3.2, 
there are three sources of prediction error: (i) $ \textsf{D}_{\text{ridge}}(\lambda_2)$,  (ii) $ \widetilde X\beta_0$ and (iii)  $\K_{\lambda_2}X(\widehat\delta-\delta_0)$. The behavior of the first two terms is determined by the behavior of the ridge estimator of the dense component $\beta_0$, and the behavior of the third term is determined by the behavior of the lasso estimator on the transformed data. 

We assume that $U\sim N(0,\sigma_u^2I_n)$ and that $X$ is fixed.    As in the lasso analysis of \cite{Bickeletal}, a key quantity is the maximal norm of the score:
$$
\Lambda= \left \|\frac{2}{n}\widetilde X'\widetilde U \right \|_{\infty}= \left \|\frac{2}{n}X'\K_{\lambda_2}U \right \|_{\infty}.
$$
Following \cite{belloni2013least}, we set the penalty level for the lasso part of lava in our theoretical development as 
\begin{equation}
\lambda_1 = c\Lambda_{1-\alpha} \ \textnormal{with} \ \Lambda_{1-\alpha} = \inf\{ l \in \mathbb{R}:  \Pr (\Lambda \leq l)  \geq 1-\alpha \}
\end{equation}
and $c>1$ a constant.  Note that \cite{belloni2013least} suggest setting $c = 1.1$ and that $\Lambda_{1-\alpha}$ is easy to approximate by simulation.

Let $S=X'X/n$ and  $\bar V_{\lambda_2}$ be the maximum  diagonal element of 
 $$V_{\lambda_2}=(S+\lambda_2I_p)^{-1}S(S+\lambda_2I_p)^{-1}\lambda_2^2.$$
Then by the union bound and Mill's inequality:
\begin{equation}
\Lambda_{1-\alpha} < \bar \Lambda_{1-\alpha} := 2\sigma_u\sqrt{\frac{\bar V_{\lambda_2} \log (2p/\alpha)}{n}}.
\end{equation}

Thus the choice $\Lambda_{1-\alpha}$ is strictly sharper than the union bound-based, classical
choice  $\bar \Lambda_{1-\alpha} $. Indeed, $\Lambda_{1-\alpha}$
is strictly smaller than $\bar \Lambda_{1-\alpha}$ even in orthogonal design cases since union bounds are not sharp.
In collinear or highly-correlated designs, it is easy to give examples where
 $\Lambda_{1-\alpha} = o(\bar \Lambda_{1-\alpha})$; see \cite{BCW-AOS2014}.  Thus, the gains from using the more refined choice
can be substantial.

We define the following design impact factor: For $\widetilde X = \K_{\lambda_2}^{1/2} X$,
$$\iota (c, \delta_0, \lambda_1,\lambda_2) := \inf_{\Delta \in \mathcal{R}(c, \delta_0, \lambda_1,\lambda_2) }
\frac{\|\widetilde X \Delta \|_2/\sqrt{n}}{\| \delta_0\|_1 - \|\delta_0 + \Delta \|_1 + c^{-1} \| \Delta \|_1}, $$
where $\mathcal{R}(c, \delta_0, \lambda_1,\lambda_2)  = \{ \Delta \in \mathbb{R}^p\setminus\{0\}: \|\widetilde X \Delta \|^2_2/n
\leq 2 \lambda_1 (\| \delta_0\|_1 - \| \delta_0 + \Delta \|_1  + c^{-1} \|\Delta\|_1)\}$ is the restricted set, and
where $\iota (c,\delta_0, \lambda_1,\lambda_2):= \infty$ if $\delta_0 = 0$.

The design impact factor generalizes the restricted eigenvalues of \cite{Bickeletal} and
and is tailored for bounding estimation errors in the prediction norm (cf. \cite{BCW-AOS2014}). Note that in the best case, when the design is well-behaved and $\lambda_2$ is a constant, we have that 
\begin{equation}\label{typical}
\iota (c, \delta_0, \lambda_1,\lambda_2) \geq 
 \frac{1}{\sqrt{\|\delta_0\|_0}}  \kappa,
\end{equation}
where $\kappa>0$ is a constant. Remarks given below provide further discussion.

The following theorem provides the deviation bounds for the lava prediction error.
\begin{theorem}[Deviation Bounds for Lava in Regression]\label{th3.3} We have that with probability $1-\alpha-\epsilon$ \begin{eqnarray*}
  \frac{1}{n}\|X\widehat\theta_{\mathrm{lava}}-X\theta_0\|_2^2&\leq& 
  \frac{2}{n}\|\K^{1/2}_{\lambda_2}X(\widehat\delta-\delta_0)\|_2^2 \|\K_{\lambda_2}\| +\frac{2}{n}\|\mathsf{D}_{\text{ridge}}(\lambda_2)\|_2^2\\
  &  \leq &  \inf_{(\delta'_0, \beta_0')' \in \mathbb{R}^{2p}: \delta_0 + \beta_0 = \theta_0}\left\{ \Big ( B_1(\delta_0)\vee B_2(\beta_0) \Big )\|\K_{\lambda_2}\| + B_3 + B_4(\beta_0) \right \},
  \end{eqnarray*} where $\|\K_{\lambda_2}\|\leq 1$ and \begin{eqnarray*}
&&   B_1(\delta_0) = \frac{2^3 \lambda_1^2}{\iota^2(c, \delta_0, \lambda_1, \lambda_2)}   \leq \frac{2^5 \sigma^2_u c^2 \bar V^2_{\lambda_2} \log (2p/\alpha)}{n \iota^2(c, \delta_0, \lambda_1, \lambda_2)},\\
&& B_2(\beta_0) =   \frac{2^5 }{n}\|\K^{1/2}_{\lambda_2} X\beta_0\|_2^2 = 2^5 \lambda_2\beta_0'S(S+\lambda_2I)^{-1}\beta_0,
\\
&& B_3 = \frac{2^2 \sigma_u^2}{n}\left [\sqrt{\tr(\PR_{\lambda_2}^2)} +  \sqrt{2} \sqrt{\|\PR^2_{\lambda_2}}\| \sqrt{\log(1/\epsilon)} \right]^2, \\
 && B_4(\beta_0) =   \frac{2^2}{n}\|\K_{\lambda_2}X\beta_0\|_2^2 = 2^2 \beta_0'V_{\lambda_2}\beta_0 \leq  2^{3}B_2(\beta_0)  \|\K_{\lambda_2}\|.
  \end{eqnarray*}
%  where $\PR_{\lambda}^2 = S^2 (S+\lambda_2I)^{-2}$. 
%  Furthermore, note that $\K_{\lambda_2}X=\lambda_2X(S+\lambda_2I)^{-1}$. Hence
% $$
 % \frac{4}{n}\|\K_{\lambda_2}X\beta_0\|_2^2=4\lambda_2^2\beta_0'(S+\lambda_2I)^{-1}S(S+%\lambda_2I)^{-1}\beta_0=4\beta_0'V\beta_0.
% $$
 %
 % $$

\end{theorem}

\begin{remark}
As noted before, the ``sparse+dense'' framework does not require the separate identification of $(\beta_0,\delta_0)$. Consequently, the prediction upper bound is the infimum over all the pairs $(\beta_0,\delta_0)$ such that $\beta_0 + \delta_0 = \theta_0$. The upper bound thus optimizes over the best ``split" of $\theta_0$ into sparse and dense parts, $\delta_0$
 and $\beta_0$.   The bound has four components. $B_1$ is a qualitatively sharp bound on the performance
of the lasso for $\K^{1/2}_{\lambda_2}$-transformed data. It involves two important factors: $\bar V_{\lambda_2}$ and 
the design impact factor $\iota(c, \delta_0,  \lambda_1,\lambda_2)$.  The term $B_3$ is the size of the impact of the noise
on the ridge part of the estimator, and it has a qualitatively sharp form as in \cite{hsu2014random}. The term $B_4$ describes the size of the bias for the ridge part of the estimator and appears to be qualitatively sharp  as in \cite{hsu2014random}.   We refer the reader to \cite{hsu2014random} for the in-depth
analysis of noise term $B_3$ and bias term $B_4$.   The term $B_2 \|\K_{\lambda_2}\|$ appearing in the bound is also related to the size of the bias resulting from ridge regularization. In examples like the Gaussian sequence model, we have
 \begin{equation}\label{BB}
 B_4(\beta_0) \lesssim B_2(\beta_0) \|\K_{\lambda_2}\| \lesssim B_4(\beta_0).
 \end{equation}
 This result holds more generally whenever $ \|\K^{-1}_{\lambda_2}\|\|\K_{\lambda_2}\| \lesssim 1$,
 which occurs if $\lambda_2$ stochastically dominates the eigenvalues of $S$ (see our supplementary material \cite{lava} for detailed derivations).  \qed
\end{remark}

\begin{remark}[Comments on Performance in Terms of Rates]
It is worth discussing heuristically two key features arising from Theorem \ref{th3.3}.

\textbf{1)} In dense models where ridge would work well, lava will work similarly to ridge.
%in dense models that have no sparse components as long as ridge would work well in those models. 
Consider any model where there is no sparse component (so $\theta_0 = \beta_0$), where the ridge-type rate $B^* = B_4(\beta_0) + B_3$ is optimal (e.g. \cite{hsu2014random}), and where (\ref{BB}) holds. In this case, 
we have $B_1(\delta_0) = 0$ since $\delta_0 =0$, and the lava performance bound reduces to
$$
B_2(\beta_0)\|\K_{\lambda_2}\| + B_3 + B_4(\beta_0) \lesssim B_4(\beta_0) + B_3 = B^*.
$$

\textbf{2)} Lava works similarly to lasso in sparse models that have no dense components
whenever lasso works well in those models. For this to hold, we need to set $\lambda_2  \gtrsim n$. Consider any model where $\theta_0 = \delta_0$ and with design such that the restricted eigenvalues $\kappa$ of \cite{Bickeletal} are bounded away from zero.  In this case, the standard lasso rate 
$$B^* =  \frac{\|\delta_0\|_0 \log (2p/\alpha)}{n \kappa^2}$$
of \cite{Bickeletal} is optimal.   
For the analysis of lava in this setting, we have that $B_2(\beta_0) = B_4(\beta_0) =0$. Moreover, we can show that 
 $B_3 \lesssim n^{-1}$ and that the design impact factor obeys (\ref{typical}) in this case.  Thus,
 $$
 B_1(\delta_0) \lesssim  \frac{\|\delta_0\|_0 \log (2p/\alpha)}{n \kappa^2} = B^*,
 $$
and $(B_1(\delta_0) \vee B_2(\beta_0))\|\K_{\lambda_2}\| + B_3 + B_4(\beta_0) \lesssim B^*$ follows due to $\|\K_{\lambda_2}\|\leq 1$.

Note that we see lava performing similarly to lasso in sparse models and performing similarly to ridge in dense models in the simulation evidence provided in the next section.  This simulation evidence is consistent with the observations made above. \qed
\end{remark}

 %\begin{remark}    [Comments on Performance in the Orthogonal Case]
 %Note that  when $\lambda_2$ is relatively large, $\bar v$ and $\kappa(\frac{2}{3}, J,\lambda_2)$ respectively reduce to $\max_{j\leq p}\frac{1}{n}\sum_{i=1}^nX_{ij}^2$ and the regular restricted eigenvalue constant as in \cite{Bickeletal}.  
% For a more general $\lambda_2$, consider the case when $S=I_p$. Then $\kappa_n(c,J,\lambda_2)=k$  and $\bar v=k^2$. Then the first term reduces to 
% $$
% \frac{128\sigma_u^2|J|_0\bar v \log p}{n\kappa(\frac{2}{3}, J,\lambda_2)} = \frac{128\sigma_u^2|J|_0k\log p}{n },\quad k=%\frac{\lambda_2}{1+\lambda_2}.
% $$   In this case, the upper bound of Theorem \ref{th3.3} becomes:
%$$
% \frac{128\sigma_u^2|J|_0k\log p}{n }+\frac{4\sigma_u^2 p(1-k)^2}{\epsilon n}+\|\beta_0\|_2^2(4k^2+8k).
%$$ Note that when $\|\beta_0\|_2=o(1)$ but $\|\theta_0\|_2^2>0$ as $p\to\infty$, the derived upper bound converges with the %``ideal choice" of $\lambda_2$. In contrast, when $p$ is large, the risk of ridge estimator $(X'X+n\lambda_rI_p)^{-1}X'Y$ does not converge regardless of the choice of $\lambda_r$.
% \end{remark}

\begin{remark}[On the design impact factor]  The definition of the design impact factor is motivated by the generalizations of the restricted eigenvalues of \cite{Bickeletal} proposed in \cite{BCW-AOS2014} to improve performance bounds for lasso in badly behaved
designs.  The concepts above are strictly more general than the usual restricted eigenvalues formulated for the transformed data. 
Let $J(\delta_0) =\{j\leq p: \delta_{0j}\neq0\}$. For any vector $\Delta \in\mathbb{R}^p$, respectively write  $\Delta_{J(\delta_0)}=\{\Delta_j: j\in J(\delta_0)\}$ and $\Delta_{J^c(\delta_0)}=\{\Delta_j: j\notin J(\delta_0)\}$. Define
$$
\mathcal{A}(c, \delta_0)=\{v\in\mathbb{R}^p\setminus \{0\}: \|\Delta_{J^c}(\delta_0) \|_1\leq (c+1)/(c-1) \|\Delta_{J(\delta_0)}\|_1\}.
$$
The restricted eigenvalue $\kappa^2(c, \delta_0,\lambda_2)$ is given by 
$$
\kappa^2(c, \delta_0,\lambda_2)=\inf_{\Delta \in \mathcal{A}(c, \delta_0)}\frac{\|\widetilde X \Delta\|_2^2/n}{\| \Delta_{J(\delta_0)}\|_2^2}=
\inf_{\Delta \in \mathcal{A}(c, \delta_0)}\frac{X'\K_{\lambda_2}X/n}{\| \Delta_{J(\delta_0)} \|_2^2}. $$
Note that  $\mathcal{A}(c, \delta_0) \subset \mathcal{R}(c, \delta_0, \lambda_1,\lambda_2)$ and that
$$
\iota (c,\delta_0, \lambda_1,\lambda_2) \geq \inf_{\Delta \in \mathcal{A}(c, \delta_0)}
\frac{\|\widetilde X \Delta \|_2/\sqrt{n}}{\|\Delta_{J(\delta_0)}\|_1} \geq \frac{1}{\sqrt{\|\delta_0\|_0}}  \kappa(c, \delta_0,\lambda_2).
$$

Now note that $X'\K_{\lambda_2}X/n=\lambda_2S(S+\lambda_2 I_p)^{-1} $. When $\lambda_2$ is relatively large, $X'\K_{\lambda_2}X/n=\lambda_2S(S+\lambda_2 I_p)^{-1}$ is approximately equal to $S$. Hence, $\kappa^2(c, \delta_0,\lambda_2)$ behaves like the usual restricted eigenvalue constant as in \cite{Bickeletal}, and we have a good bound on 
the design impact factor $\iota(c, \delta_0, \lambda_1, \lambda_2)$ as in (\ref{typical}). To understand how  $\kappa^2(c, \delta_0,\lambda_2)$ depends on $\lambda_2$ more generally, consider the special case of an orthonormal design.  In this case,  $S=I_p$ and $X'\K_{\lambda_2}X/n=kI_p$ with $k=\lambda_2/(1+\lambda_2)$. Then $\kappa^2(c, \delta_0,\lambda_2)=k$, and the design impact factor becomes  $ \sqrt{k}/\sqrt{\|\delta_0\|_0}$. 

Thus, 
the design impact factor scales like $1/\sqrt{\|\delta_0\|_0}$ when restricted eigenvalues are well-behaved, e.g. bounded away from zero.  This behavior corresponds to the best possible case.  Note that design impact factors can be well-behaved even if restricted
eigenvalues are not. For example, suppose we have two regressors that are identical.
Then $\kappa(c, \delta_0,\lambda_2) = 0$, but $\iota (c,\delta_0, \lambda_1,\lambda_2) >0$ in this case; see \cite{BCW-AOS2014}).  \qed
\end{remark}

 \section{Simulation Study}

 The lava and post-lava algorithm can be summarized as follows.
 
 \begin{itemize}
 \item[1.] Fix $\lambda_1, \lambda_2$, and define $\mathsf{P}_{\lambda_2}=X(X'X+n\lambda_2 I_p)^{-1}X'$, $\mathsf{K}_{\lambda_2}=I_n-\mathsf{P}_{\lambda_2}$.
  \item[2.] For  $\widetilde Y=\mathsf{K}_{\lambda_2}^{1/2}Y, $  and $\widetilde X=\mathsf{K}_{\lambda_2}^{1/2}X$, solve for
$$
\widehat\delta=\arg\min_{\delta \in \mathbb{R}^p} \left\{\frac{1}{n}\|\widetilde Y-\widetilde X\delta\|_2^2+\lambda_1\|\delta\|_1 \right\}.
$$
 \item[3.] Define $\widehat\beta(\delta)=(X'X+n\lambda_2I_p)^{-1}X'(Y-X\delta).$ The lava estimator is $$\widehat\theta_{\lava}=\widehat\beta(\widehat\delta) +\widehat\delta.$$
 \item[4.]  For $W=Y-X\widehat\beta(\widehat\delta)$, solve for 
$$
 \widetilde\delta= \arg\min_{\delta \in \mathbb{R}^p}\bigg\{\frac{1}{n}\|W-X\delta\|_2^2,\quad \delta_j=0\text{ if } \widehat\delta_j=0\bigg\}.
 $$
 \item[5.]    The post-lava estimator is 
 $$\widehat\theta_{\text{post-}\lava}=\widehat\beta(\widehat\delta) +\widetilde\delta.$$

 \end{itemize}

%We present results based on two methods for choosing the  tuning parameters $(\lambda_1, \lambda_2)$.  The first is to choose the tuning parameters by minimizing the SURE as given in Theorem \ref{th3.2}.  The second chooses the parameters based on $k$-fold cross-validation. The SURE formula depends on the error variance $\sigma_u^2$, which must be estimated. A conservative preliminary estimator for $\sigma_u^2$ can be obtained from an iterative method based on the regular lasso estimator; see, e.g., \cite{belloni2013least}.  On the other hand, $k$-fold cross-validation does not require a preliminary variance estimator. 

We present a Monte-Carlo analysis based on a Gaussian linear regression model: $Y=X\theta+U$, $U \mid X \sim N(0, I_n)$. The parameter $\theta$ is a $p$-vector defined as
$$\theta = (3, 0, \ldots, 0)' +  q (0, .1, \ldots, .1)',$$
where $q\geq 0$ denotes the ``size of small coefficients".  When $q$ is zero or small, $\theta$ can be well-approximated by the sparse vector $(3, 0, \ldots, 0)$. When $q$ is relatively large, $\theta$ cannot be approximated well by a sparse vector.  We set $n=100$ and $p=2n$, and compare the performance of $X\widehat\theta$ formed from one of five methods:   lasso, ridge, elastic net, lava, and post-lava.\footnote{Results with $p=n/2$, where OLS is also included, are available in supplementary material.  The results are qualitatively similar to those given here, with lava and post-lava dominating all other procedures.} %Both SURE and 5-fold cross-validation are used to select the tuning parameters for the last five methods. 
The  rows of $X$ are generated independently from a mean zero multivariate normal with covariance matrix $\Sigma$.  We present results under an independent design, $\Sigma=I$, and a factor covariance structure with
$\Sigma=LL'+I$ where the rows of $L$ are independently generated from  $N(0,I_3)$.   In the latter case, the columns of $X$ depend on three common factors.  We focus on a fixed design study, so the design $X$ is generated once and fixed throughout the replications.

To measure performance, we consider the risk measure $\mathrm{R}(\theta, \widehat \theta) = \Ep [\frac{1}{n}\| X\widehat\theta-X\theta\|_2^2]$ where the expectation $\Ep$ is conditioned on $X$.  For each estimation procedure, we report the simulation estimate of this risk measure formed by averaging over $B=100$ simulation replications. Figures \ref{f4} and 5 plot the simulation estimate of $\mathrm{R}(\theta, \widehat \theta)$ for each estimation method as a function of $q$,  the size of the ``small coefficients''.  In Figure \ref{f4}, all the tuning parameters are chosen via minimizing the SURE as defined in Theorem 3.2; and the tuning parameters are chosen by 5-fold cross-validation in Figure 5.  The SURE formula depends on the error variance $\sigma_u^2$, which must be estimated. A conservative preliminary estimator for $\sigma_u^2$ can be obtained from an iterative method based on the regular lasso estimator; see, e.g., \cite{belloni2013least}.  On the other hand, $k$-fold cross-validation does not require a preliminary variance estimator.

\begin{figure}[htbp]
\begin{center}
\includegraphics[height = 7cm, width=7.2cm]{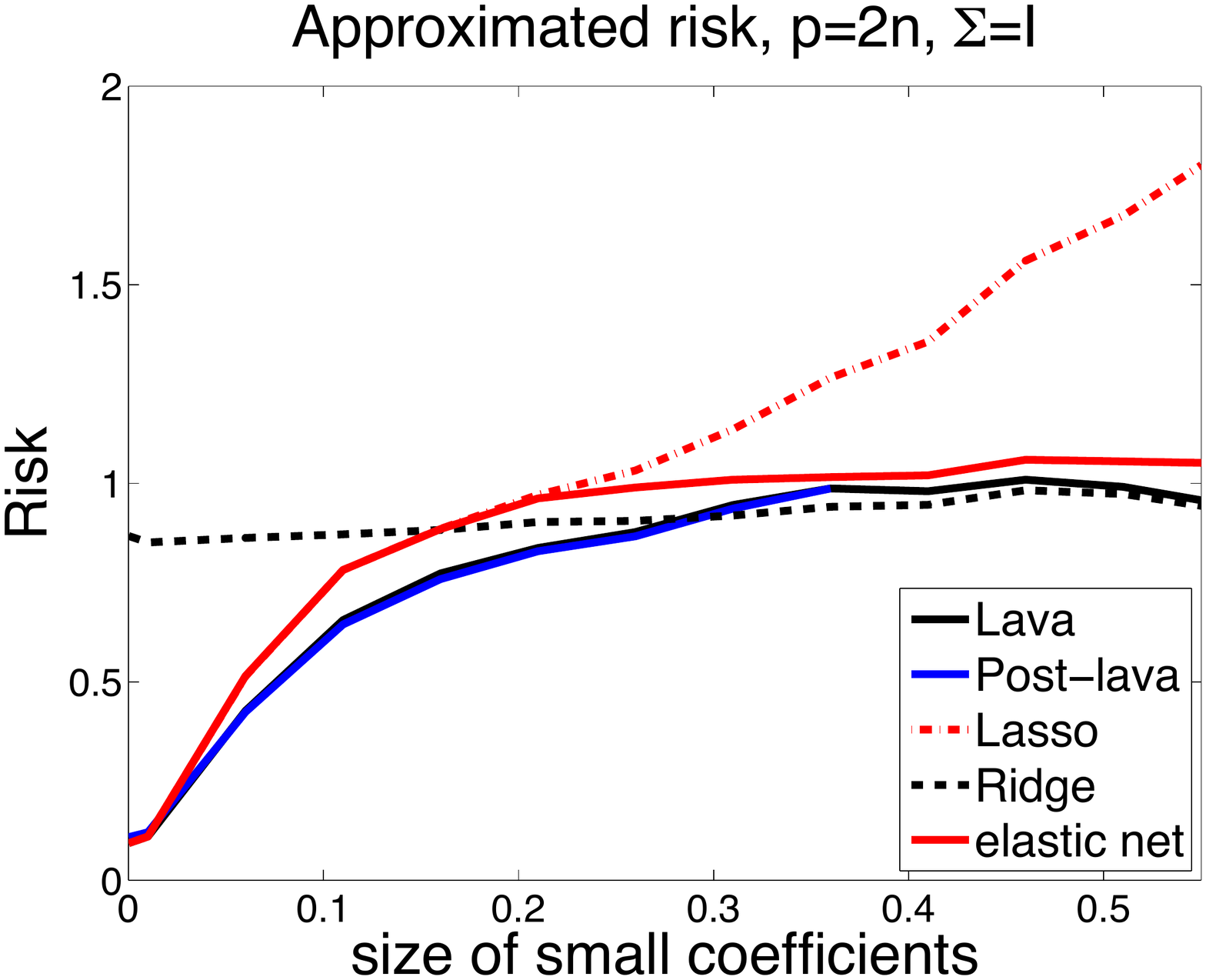}
\includegraphics[height = 7cm, width=7.2cm]{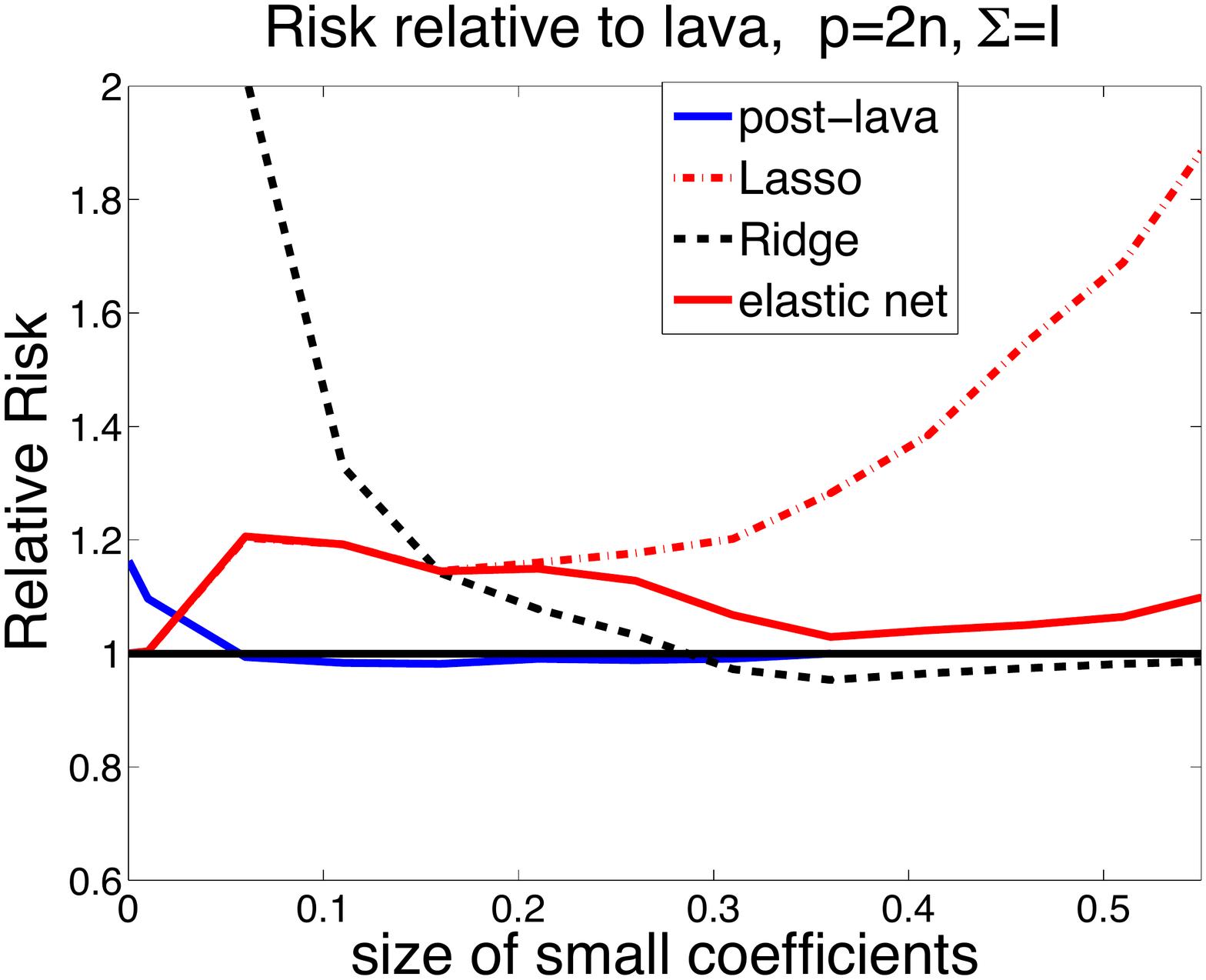}
\includegraphics[height = 7cm, width=7.2cm]{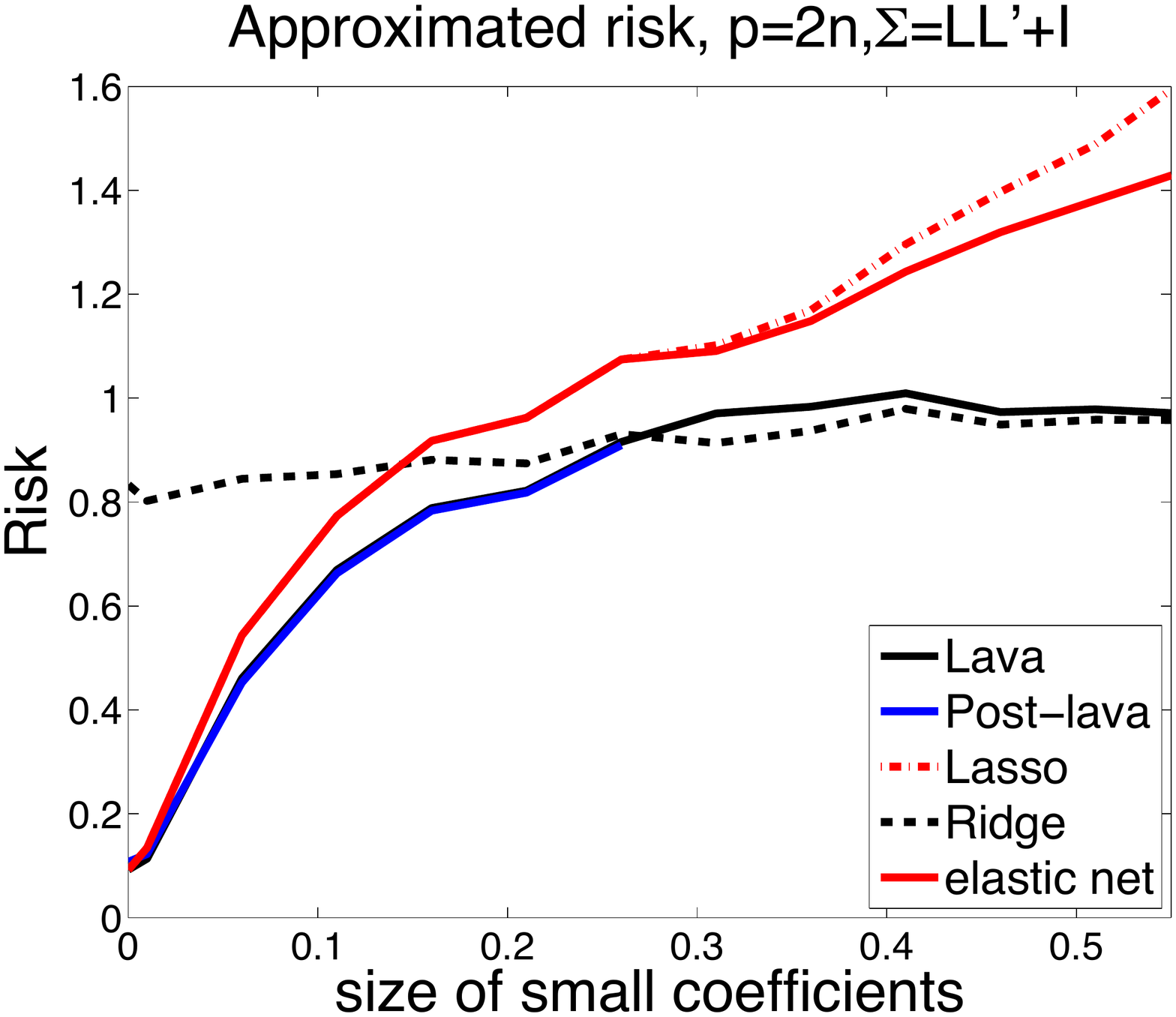}
\includegraphics[height = 7cm, width=7.2cm]{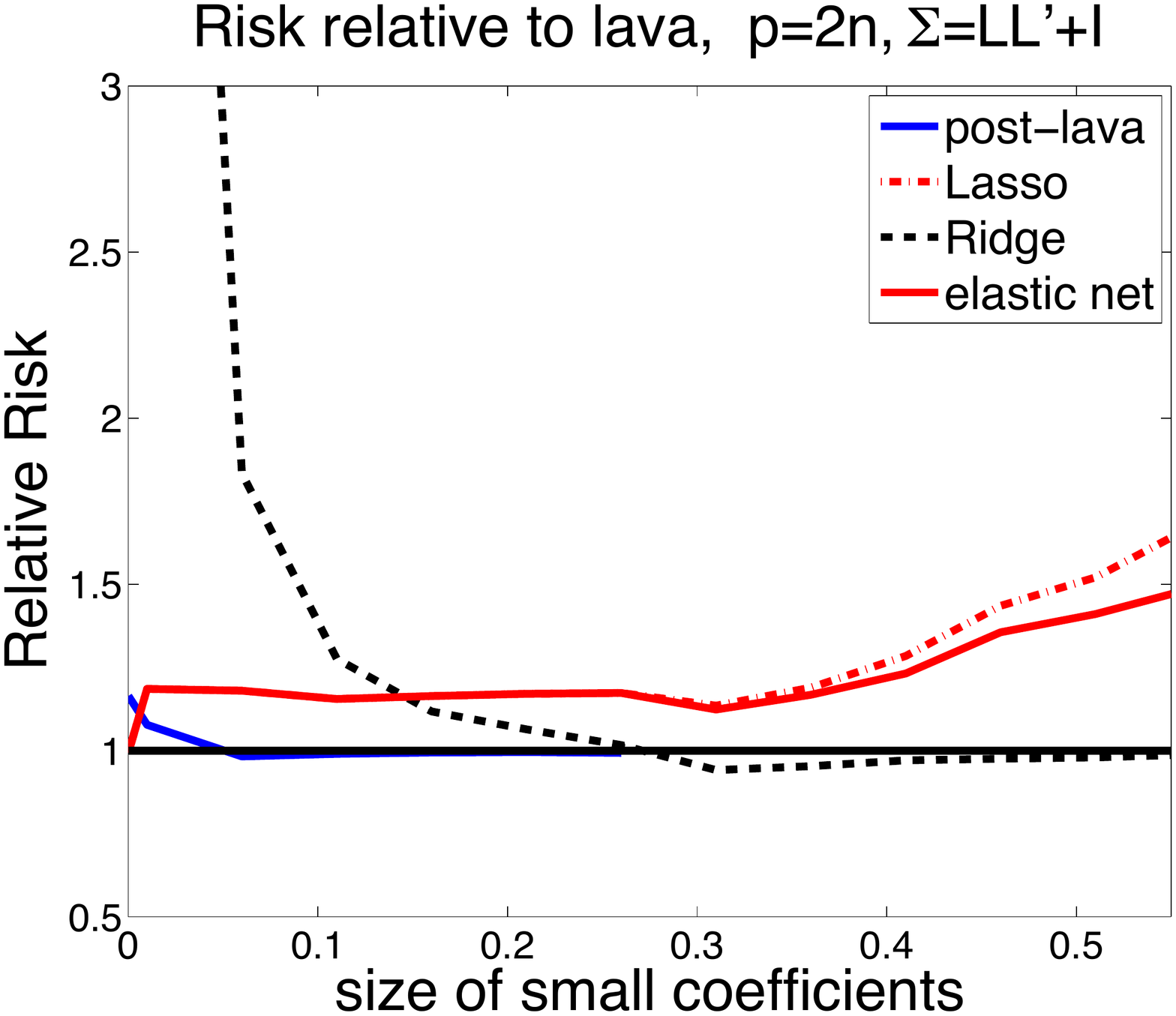}
\end{center}
\caption{\footnotesize Simulation risk comparison with tuning done by minimizing SURE.  In this figure, we report simulation estimates of risk functions of lava, post-lava, ridge, lasso, and elastic net  in a Gaussian regression model with ``sparse+dense'' signal structure over the regression coefficients.  We select tuning parameters by minimizing SURE. The size of ``small coefficients" is shown on the horizontal axis. The size of these coefficients directly corresponds to the size of the ``dense part" of  the signal, with zero corresponding to the exactly sparse case.   Relative risk plots the ratio of the risk of each estimator to the lava risk, $\mathrm{R}(\theta, \widehat \theta_e)/\mathrm{R}(\theta, \widehat \theta_{\lava})$.
}
\label{f4}
\end{figure}

\begin{figure}[htbp]
\begin{center}
\includegraphics[height = 7cm,width=7.2cm]{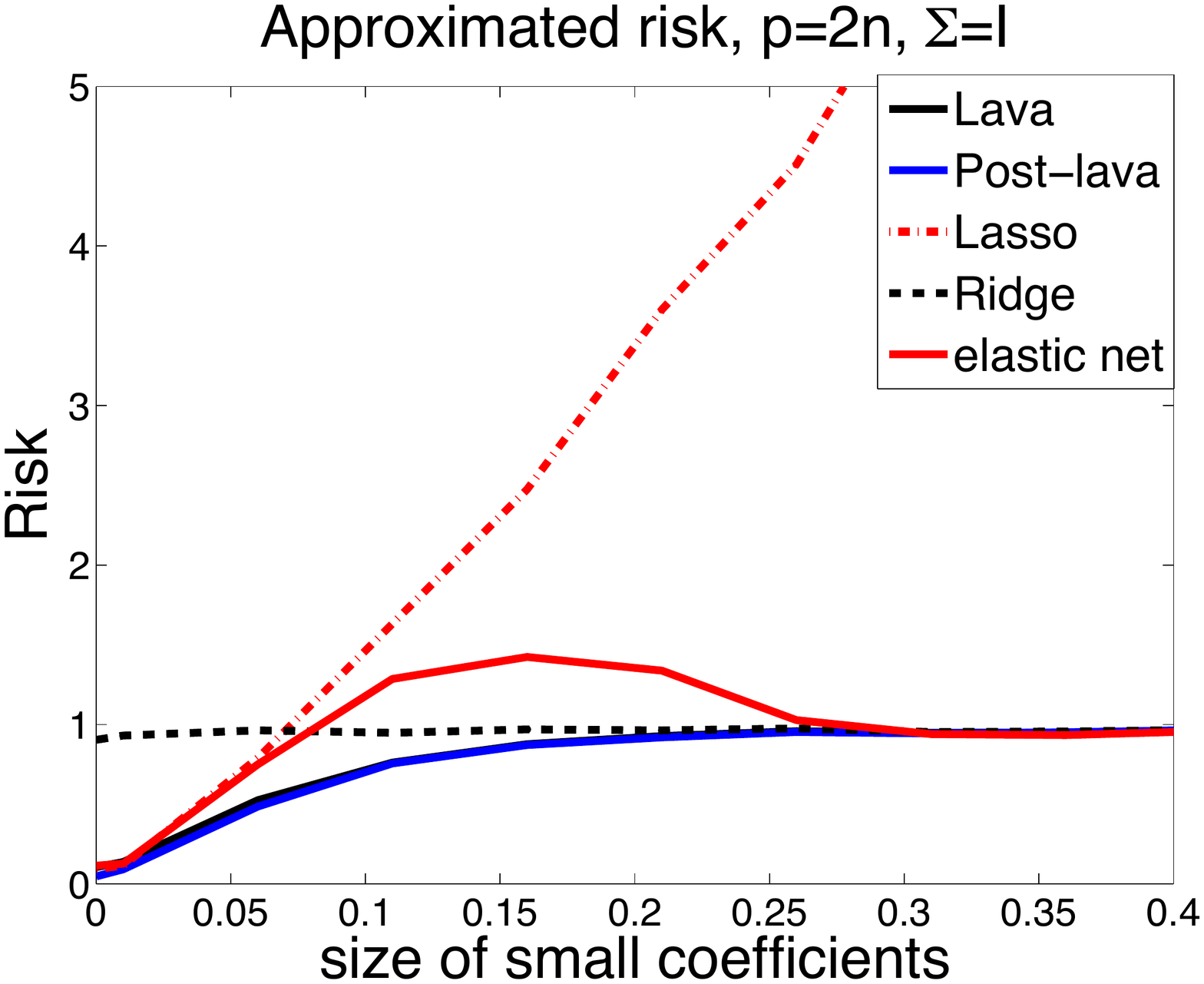}
\includegraphics[height = 7cm,width=7.2cm]{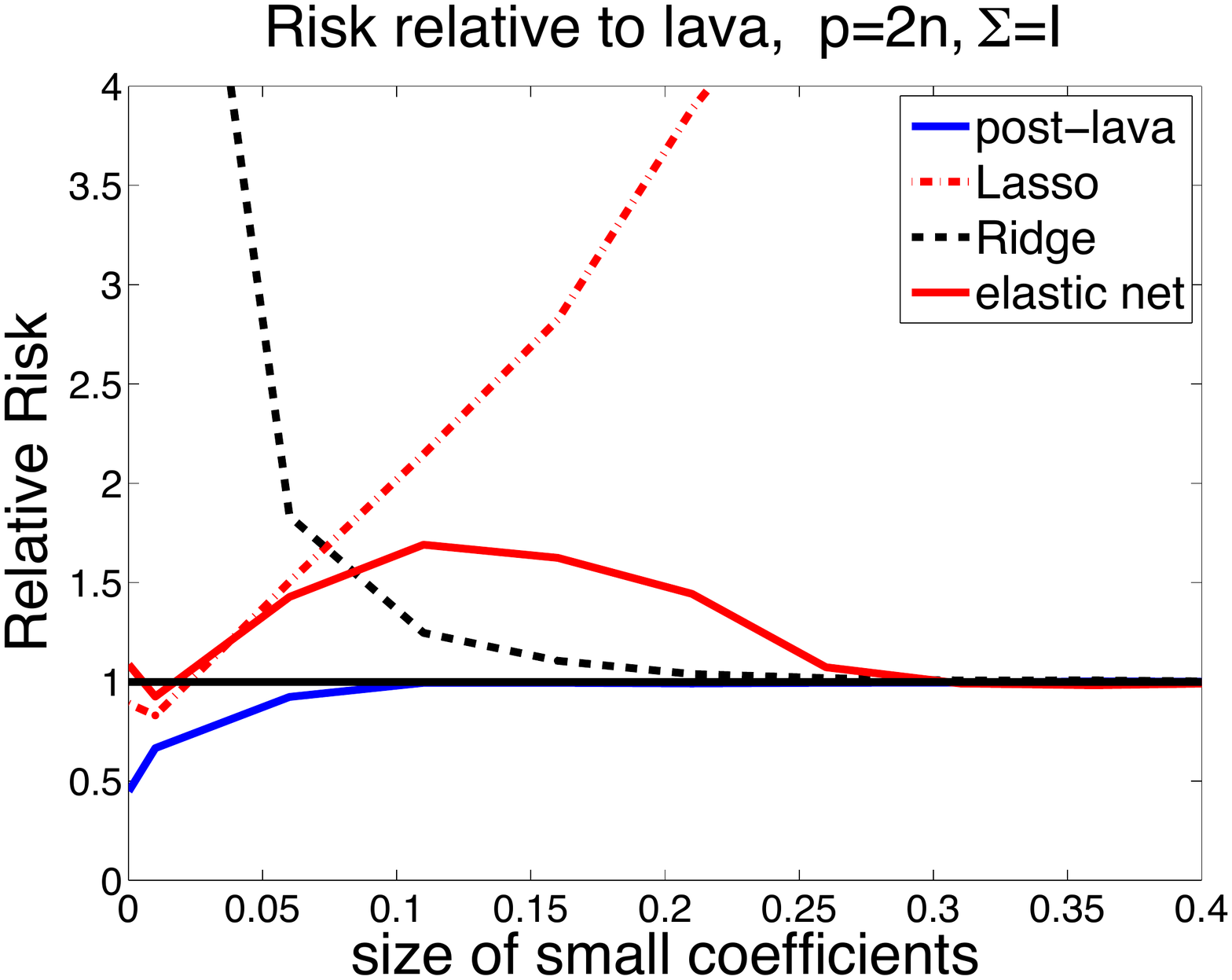}
\includegraphics[height = 7cm,width=7.2cm]{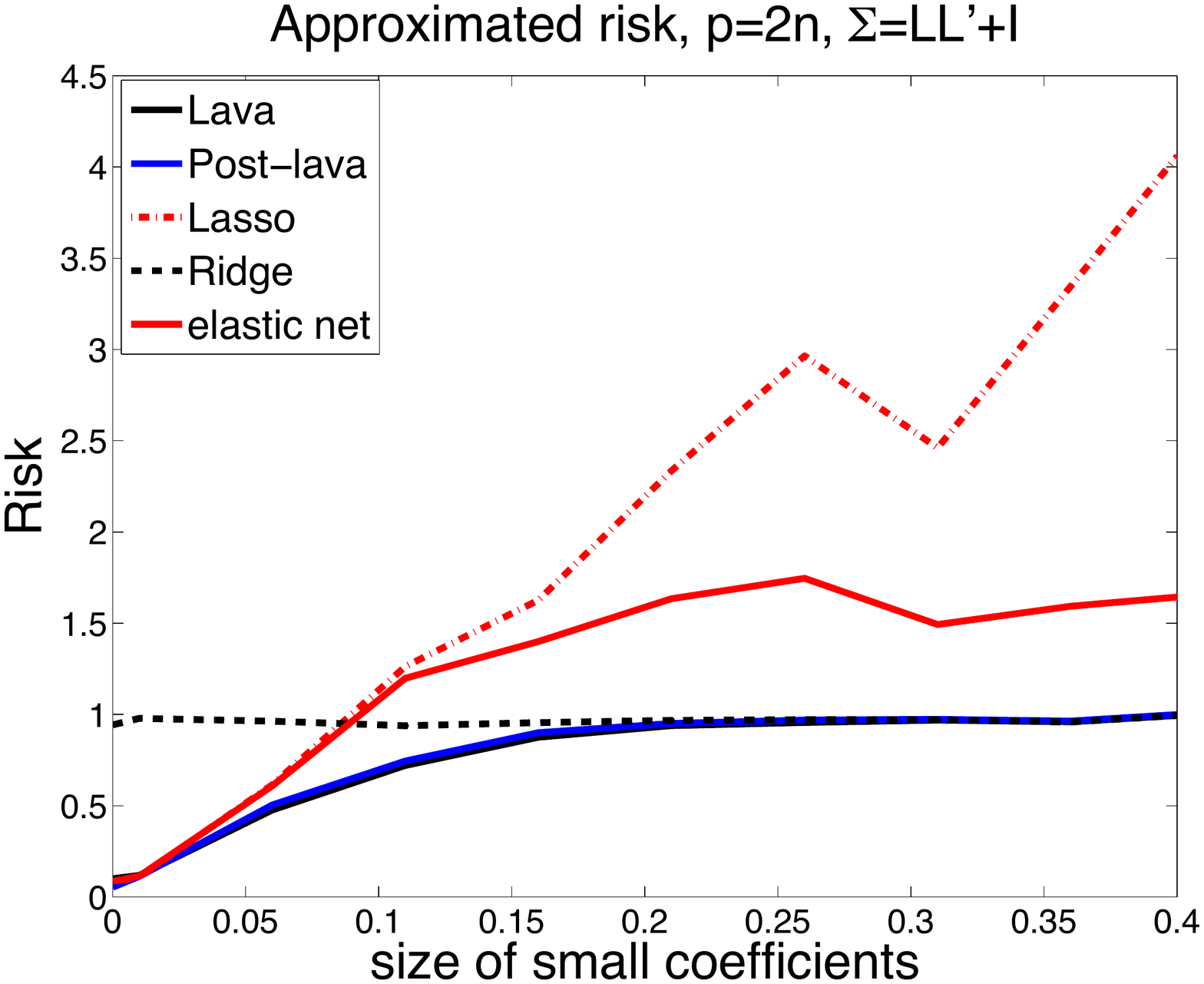}
\includegraphics[height = 7cm,width=7.2cm]{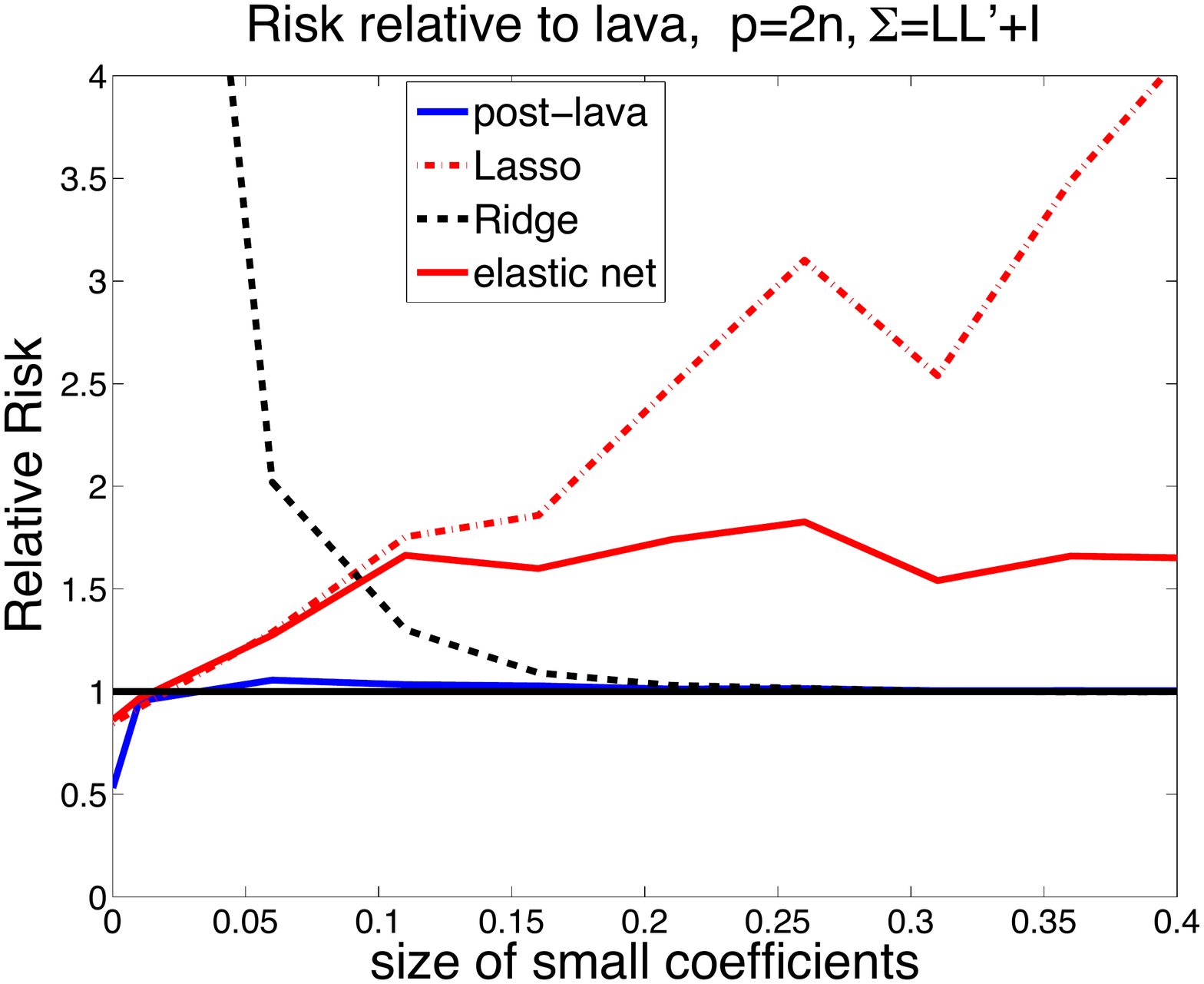}
\end{center}
\caption{\footnotesize Simulation risk comparison with tuning done by 5-fold cross-validation.  In this figure, we report simulation estimates of risk functions of lava, post-lava, ridge, lasso, and  elastic net  in a Gaussian regression model with ``sparse+dense'' signal structure over the regression coefficients.  We select tuning parameters by 5-fold cross-validation. The size of ``small coefficients" is shown on the horizontal axis. The size of these coefficients directly corresponds to the size of the ``dense part" of  the signal, with zero corresponding to the exactly sparse case.   Relative risk plots the ratio of the risk of each estimator to the lava risk, $\mathrm{R}(\theta, \widehat \theta_e)/\mathrm{R}(\theta, \widehat \theta_{\lava})$.
}
\label{f5}
\end{figure}

The comparisons are similar in both figures with lava and post-lava dominating the other procedures.  It is particularly interesting to compare the performance of lava to lasso and ridge.  The  lava and post-lava estimators perform about as well as lasso when the signal is sparse and perform significantly better than lasso when the signal is non-sparse. The lava and post-lava estimators perform about as well as ridge when the signal is dense and perform much better than ridge when the signal is sparse. When the tuning parameters are selected via cross-validation,  the post-lava  performs slightly better than the lava when the model is sparse. The gain is somewhat more apparent in the independent design. Additional simulations are presented in our supplementary material (\cite{lava}).

%All the methods, except for the lasso in the factor design of Figure 5, tend to perform no worse than least squares in very dense cases, cases where $q$ is large. 

\section{Discussion}

We propose a new method, called ``lava", which is designed specifically to achieve good prediction and estimation performance in ``sparse+dense'' models.    In such models, the high-dimensional parameter is represented as the sum of a sparse vector with a few large non-zero entries and a dense vector with many small entries.  This structure renders traditional sparse or dense estimation methods, such as lasso or ridge, sub-optimal for prediction and other estimation purposes.  The proposed approach thus complements other approaches
to structured sparsity problems such as those considered in fused sparsity estimation (\cite{tibshirani2005sparsity} and \cite{dalalyan2012fused}) and structured matrix decomposition problems (\cite{candes2011robust}, \cite{chandrasekaran2011rank}, \cite{POET}, and \cite{klopp2014robust}). 

There are a number of interesting research directions that remain to be considered. An immediate extension of the present results would be to consider semi-pivotal estimators akin to the root-lasso/scaled-lasso of \cite{BCW-Biometrika} and \cite{sun2012scaled}. For instance, we can define \begin{eqnarray*}
 \widehat\theta_{\textrm{root-lava}}&:=&\widehat\beta+\widehat\delta,\cr
 (\widehat\beta,\widehat\delta)&:=&\arg\min_{\beta,\delta, \sigma } \left \{ \frac{1}{2n\sigma^2}\|Y-X(\beta+\delta)\|_2^2+\frac{(1-a)\sigma}{2}+\lambda_2\|\beta\|_2^2+\lambda_1\|\delta\|_1 \right \}.
 \end{eqnarray*}
Thanks to the characterization of Theorem \ref{th3.1}, the method can be implemented by applying root-lasso on appropriately transformed data.  The present work could also be extended to accommodate non-Gaussian settings and settings with random designs, and it could also be extended beyond the mean regression problem to more general M- and Z- estimation problems, e.g., along the lines of \cite{negahban2009unified}.

 \appendix
 
 \section{Proofs for Section 2}
 \subsection{Proof of Lemma \ref{l2.1}} Fixing $\delta$, the solution for $\beta$ is given by $\widehat\beta(\delta)=(z-\delta)/(1+\lambda_2)$. Substituting back to the original problem, we obtain
 \begin{eqnarray*}
 d_1(z)&=&\arg\min_{\delta \in \mathbb{R}}[z-\widehat\beta(\delta)-\delta]^2+\lambda_2|\widehat\beta(\delta)|^2+\lambda_1|\delta|\cr
 &=&\arg\min_{\delta \in \mathbb{R}}k(z-\delta)^2+\lambda_1|\delta|.
 \end{eqnarray*}
 Hence $d_1(z)=(|z|-\lambda_1/(2k))_+\text{sign}(z)$, and $d_2(z)=\widehat\beta(d_1(z))=(z_1-d_1(z))(1-k)$. Consequently, $d_{\lava}(z)=d_1(z)+d_2(z)=(1-k)z+kd_1(z)$. \qed
 
 \subsection{A Useful Lemma}
  The proofs rely on the following lemma.
 \begin{lemma} \label{la.1} 
Consider the general piecewise linear function:
  $$
F(z)=\begin{cases}
 hz+d, & z>w\cr
 ez+m, &|z|\leq w \cr
 fz+g, & z<-w
 \end{cases}.
 $$
 Suppose $Z\sim N(\theta,\sigma^2)$. Then

 \begin{eqnarray*}
\Ep [F(Z)^2]& =& [\sigma^2(h^2w+h^2\theta+2dh) -\sigma^2(e^2w+e^2\theta+2me)]\phi_{\theta,\sigma}(w)\cr
&&+[\sigma^2(-e^2w+e^2\theta+2me)-\sigma^2(-f^2w+f^2\theta+2gf)]\phi_{\theta,\sigma}(-w)\cr
&&+ ((h\theta+d)^2+h^2\sigma^2)\Pr_{\theta,\sigma}(Z>w)+((f\theta+g)^2+f^2\sigma^2)\Pr_{\theta,\sigma}(Z<-w)\cr
&&+((e\theta+m)^2+e^2\sigma^2)\Pr_{\theta,\sigma}(|Z|<w).
\end{eqnarray*}

 \end{lemma}
\begin{proof}
  We first consider an expectation of the following form: for any  $-\infty\leq z_1<z_2\leq \infty$, and    $a, b\in\mathbb{R}$,  by integration by part,
\begin{eqnarray}\label{eqa.1}
&&\Ep(\theta-Z)(aZ+b)1\{z_1<Z<{z_2}\}=\sigma^2\int_{z_1}^{z_2}\frac{\theta-z}{\sigma^2}(az+b)\phi_{\theta,\sigma}(z)dz
\cr
&&=\sigma^2(az+b)\phi_{\theta,\sigma}(z)\big{|}_{z_1}^{z_2}-\sigma^2a\int_{z_1}^{z_2}\phi_{\theta,\sigma}(z)dz\cr
&&=\sigma^2[(a{z_2}+b)\phi_{\theta,\sigma}(z_2)-(a{z_1}+b)\phi_{\theta,\sigma}(z_1)]- \sigma^2a\Pr_{\theta,\sigma}(z_1<Z<z_2).
\end{eqnarray}

 This result will be useful in the following calculations. Setting $a=-1$, $b=\theta$ and $a=0$, $b=-2(\theta+c)$ respectively yields 
$$
\Ep(\theta-Z)^21\{z_1<Z<{z_2}\}=\sigma^2[(\theta-{z_2})\phi_{\theta,\sigma}(z_2)-(\theta-{z_1})\phi_{\theta,\sigma}(z_1)]+\sigma^2\Pr_{\theta,\sigma}(z_1<Z<z_2).
$$
 $$
 2\Ep(Z-\theta)(\theta+c)1\{z_1<Z<{z_2}\}=\sigma^2[-2(\theta+c)\phi_{\theta,\sigma}(z_2)+2(\theta+c)\phi_{\theta,\sigma}(z_1)].
 $$
 Therefore,  for any constant $c$,
\begin{eqnarray}\label{eqa.2}
 \Ep(Z+c)^21\{z_1<Z<{z_2}\}&=&\Ep(\theta-Z)^21\{z_1<Z<{z_2}\}+  (\theta+c)^2\Pr_{\theta,\sigma}(z_1<Z<{z_2}) \cr
&&+2\Ep(Z-\theta)(\theta+c)1\{z_1<Z<{z_2}\}\cr
&=&\sigma^2(z_1+\theta+2c)\phi_{\theta,\sigma}(z_1)-\sigma^2(z_2+\theta+2c)\phi_{\theta,\sigma}(z_2)
 \cr
 &&+((\theta+c)^2+\sigma^2)\Pr_{\theta,\sigma}(z_1<Z<z_2).
\end{eqnarray}
If none of $h,e,f$ are zero, by setting $z_1=w, z_2=\infty$, $c=d/h$;   $z_1=-\infty, z_2=-w,  c=g/f$ and $z_1=-w, z_2=w, c=m/e$ respectively,  we have
\begin{eqnarray}\label{eqa.3}
&& \Ep(hZ+d)^21\{Z>w\} =\sigma^2(h^2w+h^2\theta+2dh)\phi_{\theta,\sigma}(w) \cr
&& \hspace{1.5in}+ ((\theta h+d)^2+\sigma^2h^2)\Pr_{\theta,\sigma}(Z>w),\cr
&& \Ep(fZ+g)^21\{Z<-w\}=-\sigma^2(-wf^2+\theta f^2+2gf)\phi_{\theta,\sigma}(-w) \\
&& \hspace{1.5in} +((\theta f+g)^2+\sigma^2f^2)\Pr_{\theta,\sigma}(Z<-w),\cr
&& \Ep(eZ+m)^21\{|Z|<w\}=\sigma^2(-we^2+\theta e^2+2me)\phi_{\theta,\sigma}(-w)\cr
&& \hspace{1.5in} -\sigma^2(we^2+\theta e^2+2me)\phi_{\theta,\sigma}(w)
\cr
&&  \hspace{1.5in}  +((\theta e+m)^2+\sigma^2 e^2)\Pr_{\theta,\sigma}(|Z|<w).
\end{eqnarray}
If any of  $h,e,f$ is zero, for instance, suppose $h=0$,  then
$
\Ep(hZ+d)^21\{Z>w\} =d^2\Pr_{\theta,\sigma}(Z>w),
$ which can also be written as the first equality of (\ref{eqa.3}). Similarly, when either $e=0$ or $f=0$, (\ref{eqa.3}) still holds.

Therefore, summing up the   three terms of (\ref{eqa.3}) yields the desired result. \end{proof}

\subsection{Proof of Theorem \ref{th2.1}}   Recall that 
$\widehat\theta_{\text{lava}}=(1-k)Z+k\sh_{\lasso}(Z)$ is  a weighted average of $Z$ and the soft-thresholded estimator
 with shrinkage parameters $\lambda_1/(2k)$ and $k=\lambda_2/(1+\lambda_2)$.   Since $\sh_{\lasso}(Z)$ is a soft-thresholding estimator, results from \cite{donoho1995adapting} give that
$$
\Ep[(Z-\theta)\sh_{\lasso}(Z)]=\sigma^2\Pr_{\theta,\sigma}(|Z|>\lambda_1/(2k)).
$$
Therefore,  for $w=\lambda_1/(2k)$,
 \begin{eqnarray}\label{eq2.4}
2\Ep[(Z-\theta)d_{\lava}(Z)]=2(1-k)\sigma^2 +2k\sigma^2\Pr_{\theta,\sigma}(|Z|>w).
 \end{eqnarray}
Next we verify that 
    \begin{equation}
    \begin{array}{ll}
\Ep(Z-\sh_{\lava}(Z))^2&=-k^2(w  +\theta)\phi_{\theta,\sigma}(w)\sigma^2+k^2( \theta-w)\phi_{\theta,\sigma}(-w)\sigma^2 \cr
&+(\lambda_1^2/4)\Pr_{\theta,\sigma}(|Z|>w) +k^2(\theta^2+\sigma^2 )\Pr_{\theta,\sigma}(|Z|<w).
\end{array}
\end{equation}
 By definition,
 \begin{eqnarray}  \sh_{\lava}(z)-z&=& \begin{cases}
 -\lambda_1/2,& z>\lambda_1/(2k),\\
-kz, & -\lambda_1/(2k)<z\leq\lambda_1/(2k),\\
  \lambda_1/2,& z<-\lambda_1/(2k).
 \end{cases}
\end{eqnarray}
Let $F(z)=\sh_{\lava}(z)-z$. The claim then follows from applying Lemma \ref{la.1} by setting $h=f=m=0$, $d=-\lambda_1/2$, $e=-k$, $g=\lambda_1/2$, and  $w=\lambda_1/(2k)$.   

Hence \begin{eqnarray*}
\Ep(Z-d_{\lava}(Z))^2&=&-k^2(w  +\theta)\phi_{\theta,\sigma}(w)\sigma^2+k^2( \theta-w)\phi_{\theta,\sigma}(-w)\sigma^2 \cr
&&+(\lambda_1^2/4)\Pr_{\theta,\sigma}(|Z|>w) +k^2(\theta^2+\sigma^2 )\Pr_{\theta,\sigma}(|Z|<w).
\end{eqnarray*}

The risk of lasso is obtained from setting $\lambda_2=\infty$ and $\lambda_1=\lambda_l$ in the lava risk.
The risk of ridge is obtained from setting $\lambda_1 =\infty$ and $\lambda_2 = \lambda_r$ in the lava risk.

As for the risk of post-lava, note that 
 \begin{eqnarray*} 
\sh_{\textrm{post-lava}}(z)-\theta=\begin{cases}
 z-\theta,& |z|>\lambda_1/(2k)\\
(1-k)z-\theta, & |z|\leq\lambda_1/(2k).
 \end{cases}
 \end{eqnarray*}
Hence applying Lemma \ref{la.1} to $F(z) = \sh_{\textrm{post-lava}}(z)-\theta$,
i.e. by setting $h=f=1, e=1-k$ and $d=m=g=-\theta$, we obtain:
 \begin{eqnarray*} 
R(\theta,\widehat\theta_{\text{post-lava}})&=&\sigma^2[-k^2w  +2kw -k^2\theta ]\phi_{\theta,\sigma}(w)+\sigma^2 [-k^2w+2kw+k^2\theta]\phi_{\theta,\sigma}(-w)\cr
&&+ \sigma^2\Pr_{\theta,\sigma}(|Z|>w)+(k^2\theta^2+(1-k)^2\sigma^2)\Pr_{\theta,\sigma}(|Z|<w).
\end{eqnarray*} 

Finally, the elastic net shrinkage is given by$$\sh_{\mathrm{enet}}(z)=\frac{1}{1+\lambda_2}(|z|-\lambda_1/2)_+\text{sgn}(z).$$
 The risk of elastic net then follows from Lemma \ref{la.1} by setting 
 $F(z) = \sh_{\mathrm{enet}}(z)- \theta$, 
 $w=\lambda_1/2, h=f=1/(1+\lambda_2), e=0$, $d=-\lambda_1/(2(1+\lambda_2))-\theta$ and $g=\lambda_1/(2(1+\lambda_2))-\theta$.   \qed

%\subsection{Risk of Elastic Net}The following result gives the risk for elastic-net, omitted in the main text.

%    \begin{lemma}   For $h=1/(1+\lambda_2)$, $d=-\lambda_1/(2(1+\lambda_2))-\theta$ and $g=\lambda_1/(2(1+\lambda_2))-\theta$,
  %  \begin{eqnarray*}
%R(\theta,\widehat\theta_{\mathrm{enet}})&=& \sigma^2(h^2\lambda_1/2+h^2\theta+2dh) \phi_{\theta,\sigma}(\lambda_1/2)\cr
%&&-\sigma^2(-h^2\lambda_1/2+h^2\theta+2gh)\phi_{\theta,\sigma}(-\lambda_1/2)+\theta^2\Pr_{\theta,\sigma}(|Z|<\lambda_1/2)\cr
%&&+ ((h\theta+d)^2+h^2\sigma^2)\Pr_{\theta,\sigma}(Z>\lambda_1/2)\\
%&& +((h\theta+g)^2+h^2\sigma^2)\Pr_{\theta,\sigma}(Z<-\lambda_1/2).
%\end{eqnarray*}
   % \end{lemma}

%By definition,
 %\begin{eqnarray}\label{eq2.2}
 %\widehat\theta_{\text{lava}}(z)-\theta&=& \begin{cases}
 %z-\lambda_1/2-\theta,& z>\lambda_1/(2k)\\
 %z/(1+\lambda_2)-\theta, & -\lambda_1/(2k)<z\leq\lambda_1/(2k),\\
 %z+ \lambda_1/2-\theta,& z<-\lambda_1/(2k).
 %\end{cases}
%\end{eqnarray}
%The risk function of lava then follows immediately from applying Lemma \ref{la.1} onto $F(z)=\widehat\theta_{\text{lava}}(z)-\theta$, 
  %by setting $h=f=1$, $e=1/(1+\lambda_2)$,  $d=-\lambda_1/2-\theta$, $m=-\theta$, $g=\lambda_1/2-\theta$, and $w=\lambda_1/(2k)$.

 \subsection{Proof of Lemma \ref{th2.2}}
 For $\widehat\theta_{\text{lava}} =(\widehat\theta_{\lava,j})_{j=1}^p$, 
 we have $$\Ep\|\widehat\theta_{\text{lava}}-\theta\|_2^2=\sum_{j=1}^pR(\theta_j,\widehat\theta_{\lava,j}).$$ We now bound $R(\theta_j,\widehat\theta_{\lava,j})$ uniformly over $j=1,...,p$. We
 have that 
 $$
 \widehat\theta_{\text{lava},j}-\theta_j= (1-k) Z_j + k \widehat \delta_j - \theta_j = 
 [(1-k)(Z_j-\delta_j)-\beta_j]+k(\widehat\delta_j-\delta_j),\quad k=\frac{\lambda_2}{1+\lambda_2},
 $$
 where $\widehat\delta_j= \sh_{\lasso}(Z_j)$ with penalty level $\lambda_1$. Hence
 \begin{eqnarray*}
 \Ep( \widehat\theta_j-\theta_j)^2 && =\underbrace{\Ep[(1-k)(Z_j-\delta_j)-\beta_j]^2}_{I}+\underbrace{k^2\Ep(\widehat\delta_j-\delta_j)^2}_{II} \\
&& +\underbrace{2k\Ep\{(\widehat\delta_j-\delta_j)[(1-k)(Z_j-\delta_j)-\beta_j]\}}_{III}
 \end{eqnarray*}

 \textit{Bounding I}. Note that $Z_j-\delta_j\sim N(\beta_j,\sigma^2)$. The first term $G_j\equiv (1-k)(Z_j-\delta_j)-\beta_j$ is thus the bias of a ridge estimator, with $\Ep G_j^2=(1-k)^2\sigma^2+k^2\beta_j^2$.

  \textit{Bounding II}.   Note that $ \lambda_1= 2\sigma{\Phi^{-1}(1-c/(2p))}$. By Mill's ratio inequality, 		 as long as  $ \frac{2p}{\pi c^2}\geq {\log p}\geq \frac{1}{\pi c^2}$, $2 \sqrt{2\log p}>\lambda_1/\sigma>2 \sqrt{\log p}$. In addition,
  $$
\Ep(\widehat\delta_j-\delta_j)^2=\Ep(\widehat\delta_j-\theta_j)^2+\beta_j^2+2\Ep(\widehat\delta_j-\theta_j)\beta_j.
  $$
Since  $Z_j\sim N(\theta_j,\sigma^2)$, by Theorem \ref{th2.1} with $\lambda_l=\lambda_1/k> 2\sigma\sqrt{\log p}$ (since $k\leq1$),
 \begin{eqnarray*}
\Ep(\widehat\delta_j-\theta_j)^2& =&  -(\lambda_l/2  +\theta_j)\phi_{\theta_j,\sigma}(\lambda_l/2)\sigma^2+( \theta_j-\lambda_l/2)\phi_{\theta_j,\sigma}(-\lambda_l/2)\sigma^2 \cr
&&+(\lambda_l^2/4+\sigma^2)\Pr_{\theta_j,\sigma}(|Z_j|>\lambda_l/2) +\theta_j^2\Pr_{\theta_j,\sigma}(|Z_j|<\lambda_l/2)\cr
&\leq_{(1)}&  (\lambda_l^2/4+\sigma^2) \frac{\sigma}{\sqrt{2\pi}(\lambda_l/2-\theta_j)}e^{-(\lambda_l/2-\theta_j)^2/2\sigma^2} \cr
&&+(\lambda_l^2/4+\sigma^2)\frac{\sigma}{\sqrt{2\pi}(\lambda_l/2+\theta_j)} e^{-\frac{(\lambda_l/2+\theta_j)^2}{2\sigma^2}} 
+\theta_j^2\Pr_{\theta_j,\sigma}(|Z_j|<\lambda_l/2)\cr
&\leq_{(2)}& 2\times (\lambda_l^2/4+\sigma^2) \frac{4\sigma}{\sqrt{2\pi}\lambda_l}e^{-\lambda_l^2/(32\sigma^2)}    + \theta_j^2 \cr
&\leq_{(3)}&    \frac{4 \lambda_l\sigma^2}{\sqrt{2\pi}\sigma}e^{-\lambda_l^2/(32\sigma^2)}    + \theta_j^2\cr
&\leq_{(4)}&      \frac{4 \sigma^2}{\sqrt{2\pi}}e^{-\lambda_l^2/(64\sigma^2)}  +\theta_j^2\leq\frac{4\sigma^2}{\sqrt{2\pi}p^{(1/4)^2}}+\theta_j^2
\end{eqnarray*}  
 where (1) follows from the Mill's ratio inequality: $\int_x^{\infty}e^{-t^2/2}dt\leq x^{-1}e^{-x^2/2}$ for $x\geq 0$.  Also note that $\lambda_l/2\pm \theta_j>0$ since $\|\theta\|_{\infty}<M$ and $\sigma\sqrt{\log p}>2M$. Hence we can apply the Mill's ratio inequality respectively on $\Pr_{\theta_j,\sigma}(Z_j>\lambda_l/2) $ and $\Pr_{\theta_j,\sigma}(Z_j<-\lambda_l/2) $. In addition, the first two terms  on the right hand side are negative.
  (2) is due to $\lambda_l/2\pm\theta_j>\lambda_l/4$ since   $|\theta_j|<M$ and $ \sigma\sqrt{\log p}>2M$. 
  (3) follows since $4\sigma^2\leq \lambda_l^2$  when $p\geq e$. Finally,  for any $a>0$, and any $x>1+a^{-1}$, $ax^2>\log x$. Set $a=64^{-1}$; when $ \sqrt{\log p}>33$, $\log(\lambda_l/\sigma)<\lambda_l^2/(64\sigma^2)$. Hence $  \frac{ \lambda_l}{\sigma}e^{-\lambda_l^2/(32\sigma^2)} \leq e^{-\lambda_l^2/(64\sigma^2)}$,
   which gives (4). Therefore,
 $$
 \sum_{j=1}^p\Ep(\widehat\delta_j-\theta_j)^2\leq \frac{4p\sigma^2}{\sqrt{2\pi}p^{(1/4)^2}}+ \|\theta\|_2^2.
 $$

On the other hand, applying   (\ref{eqa.1}) and by the same arguments as above,   uniformly for $j=1,...,p$, 
 \begin{eqnarray*}|\Ep\widehat\delta_j|&=&|\sigma^2[\phi_{\theta_j,\sigma}(\lambda_l/2)-\phi_{\theta_j,\sigma}(-\lambda_l/2)]+(\theta_j-\lambda_l/2)\Pr_{\theta_j,\sigma}(Z_j>\lambda_l/2)\cr
 && +(\theta_j+\lambda_l/2)\Pr_{\theta_j,\sigma}(Z_j<-\lambda_l/2)|\cr
 &\leq&\frac{8\sigma}{\sqrt{2\pi}}e^{-\lambda_l^2/(32\sigma^2)} \leq   \frac{8\sigma}{\sqrt{2\pi}} \frac{1}{p^{1/8}}.
\end{eqnarray*}  
Moreover, $\|\theta\|_2^2+\|\beta\|_2^2-2\sum_{j=1}^p\theta_j\beta_j=\|\delta\|_2^2$. Hence
\begin{eqnarray*}  
\sum_{j=1}^p\Ep(\widehat\delta_j-\delta_j)^2&\leq& \sum_{j=1}^p\Ep(\widehat\delta_j-\theta_j)^2+\|\beta\|_2^2-2\sum_{j=1}^p\theta_j\beta_j+2\sum_{j=1}^p\Ep\widehat\delta_j\beta_j\cr
&\leq& \|\delta\|_2^2+ \frac{4p\sigma^2}{\sqrt{2\pi}p^{1/16}}+\|\beta\|_{\infty}\frac{10\sigma p}{\sqrt{2\pi}p^{1/8}}.
\end{eqnarray*}

  \textit{Bounding III}. By \cite{donoho1995adapting},  $
\Ep[(\theta_j-Z_j)\widehat\delta_j]=-\sigma^2\Pr_{\theta_j,\sigma}(|Z_j|>\lambda_1/(2k)).
$ Hence  (note that $\Ep(Z_j-\theta_j)=0$)
 \begin{eqnarray*}
\Ep[(\widehat\delta_j-\delta_j)(Z_j-\delta_j)]&=& \Ep[\widehat\delta_j(Z_j-\theta_j)]+\Ep[\widehat\delta_j-\delta_j)\beta_j]\cr
&=&\sigma^2\Pr_{\theta_j,\sigma}(|Z_j|>\lambda_1/(2k))+\beta_j\Ep\widehat\delta_j-\delta_j\beta_j\cr
&\leq &2\sigma^2\frac{4\sigma}{\sqrt{2\pi}\lambda_l} e^{-\lambda_l^2/(32\sigma^2)} +\beta_j\Ep\widehat\delta_j-\delta_j\beta_j  \cr
&\leq&   \frac{4\sigma^2}{\sqrt{2\pi}}\frac{1}{\sqrt{\log p}}\frac{1}{p^{1/8}}+\beta_j\Ep\widehat\delta_j-\delta_j\beta_j,
%&\leq&\frac{4\sigma^2}{\sqrt{2\pi\log p}p^{1/8}} + \|\beta\|_{\infty} \frac{5\sigma}{\sqrt{2\pi}} \frac{1}{p^{1/8}}-\delta_j\beta_j \cr
%&\leq& \frac{\sigma}{\sqrt{2\pi}}\frac{1}{p^{(1/8)^2}}(9\|\beta\|_\infty+8\sigma)-\delta_j\beta_j\cr
%&\leq& \frac{\sigma}{\sqrt{2\pi}}\frac{7M}{p^{1/8}}-\delta_j\beta_j,
\end{eqnarray*}  
  implying
  \begin{eqnarray*}
 &&\sum_{j=1}^p\Ep\{(\widehat\delta_j-\delta_j)[(1-k)(Z_j-\delta_j)-\beta_j]\}\cr
 &=&(1-k)\sum_{j=1}^p\Ep[(\widehat\delta_j-\delta_j)(Z_j-\delta_j)]-\sum_{j=1}^p\Ep\widehat\delta_j\beta_j+\sum_{j=1}^p\delta_j\beta_j\cr
 %&\leq&(1-k) \frac{\sigma}{\sqrt{2\pi}}\frac{7Mp}{p^{1/8}}-(1-k)\sum_{j=1}^p\delta_j\beta_j-\sum_{j=1}^p\Ep\widehat\delta_j\beta_j+\sum_{j=1}^p\delta_j\beta_j\cr
 &\leq& (1-k)   \frac{4\sigma^2}{\sqrt{2\pi}}\frac{p}{\sqrt{\log p}}\frac{1}{p^{1/8}}+(1-k)\sum_{j=1}^p\beta_j\Ep\widehat\delta_j-\sum_{j=1}^p\Ep\widehat\delta_j\beta_j+k\sum_{j=1}^p\delta_j\beta_j\cr
 %&\leq& (1-k)  \frac{\sigma}{\sqrt{2\pi}}\frac{1}{p^{(1/8)^2-1}}(9\|\beta\|_\infty+8\sigma)+k\sum_{j=1}^p\delta_j\beta_j+\|\beta\|_{\infty}  \frac{9\sigma}{\sqrt{2\pi}} \frac{1}{p^{(1/8)^2-1}}\cr
 &=& (1-k)   \frac{4\sigma^2}{\sqrt{2\pi}}\frac{p}{\sqrt{\log p}}\frac{1}{p^{1/8}}-k\sum_{j=1}^p\beta_j\Ep\widehat\delta_j +k\sum_{j=1}^p\delta_j\beta_j\cr
 &\leq& (1-k)   \frac{4\sigma^2}{\sqrt{2\pi}}\frac{p}{\sqrt{\log p}}\frac{1}{p^{1/8}}+k\|\beta\|_{\infty}p \frac{8\sigma}{\sqrt{2\pi}} \frac{1}{p^{1/8}} +k\sum_{j=1}^p\delta_j\beta_j\cr
  &\leq&  M\frac{9p\sigma}{\sqrt{2\pi}} \frac{1}{p^{1/8}} +k\sum_{j=1}^p\delta_j\beta_j \quad (\text{since } M^2\log p>16\sigma^2,  k<1, \|\beta\|_{\infty}<M).
\end{eqnarray*}  
Summarizing, we obtain (note that $\|\theta\|_2^2\leq \|\beta\|_2^2+3sM^2$)
  \begin{eqnarray*}
\sum_{j=1}^p\Ep( \widehat\delta_j-\theta_j)^2&\leq &\underbrace{p(1-k)^2\sigma^2+k^2\|\beta\|_2^2}_{I}+\underbrace{k^2 \|\delta\|_2^2+k^2 \frac{4p\sigma^2}{\sqrt{2\pi}p^{1/16}}+k^2\|\beta\|_{\infty}\frac{10\sigma p}{\sqrt{2\pi}p^{1/8}}}_{II}\cr
&&+\underbrace{2k M\frac{9p\sigma}{\sqrt{2\pi}} \frac{1}{p^{1/8}} +2k^2\sum_{j=1}^p\delta_j\beta_j}_{III}\cr
&=&p(1-k)^2\sigma^2+k^2\|\beta+\delta\|_2^2+k^2 \frac{4p\sigma^2}{\sqrt{2\pi}p^{1/16}}+(10k^2+18k)\frac{M\sigma p}{\sqrt{2\pi}p^{1/8}}\cr
&\leq &p(1-k)^2\sigma^2+k^2(\|\beta\|_2^2+3sM^2) + \frac{4p\sigma^2}{\sqrt{2\pi}p^{1/16}}+\frac{28M\sigma p}{\sqrt{2\pi}p^{1/8}}
\end{eqnarray*}  
Finally, due to $\Ep\|Z-\theta\|_2^2=p\sigma^2$,  we have
  \begin{eqnarray*}
\frac{ \Ep\| \widehat\theta_{\text{lava}}-\theta\|^2}{\Ep\|Z-\theta\|^2}&\leq&  (1-k)^2 +\frac{ k^2}{p\sigma^2}\|\beta\|_2^2+\frac{ k^2}{p\sigma^2}3sM^2+ \frac{4 }{\sqrt{2\pi}p^{1/16}}+\frac{28M}{\sigma\sqrt{2\pi}p^{1/8}}\cr
%&\leq& (1-k)^2+\frac{k^2\|\beta\|_2^2}{p\sigma^2}+\frac{3sM^2}{p\sigma^2}+\frac{1}{\sqrt{2\pi}}\frac{2}{p^{1/64}}(    \frac{18M}{\sigma}+ 8)\cr
&=& 1-k+\frac{3sM^2}{p\sigma^2}+ \frac{4 }{\sqrt{2\pi}p^{1/16}}+\frac{28M}{\sigma\sqrt{2\pi}p^{1/8}}\end{eqnarray*}  
where the last equality is due to $(1-k)^2+\frac{k^2\|\beta\|_2^2}{p\sigma^2}=1-k$ for $k=\frac{\sigma^2p}{\|\beta\|_2^2+\sigma^2p}.$

\subsection{Proof of Theorem \ref{th2.3}}  The first result follows from equation (\ref{eq2.4}) in the proof of Theorem 2.1; the second result follows directly from (\ref{eq2.5}).
 \qed
%$  \Ep[(Z_j-\theta_j)\widehat\delta_j]=(1-k)\Ep[(Z_j-\theta_j)Z_j]+  k\Ep[(Z_j-\theta)\widehat\delta_j]$and  \cite{donoho1995adapting}. In addition, 
% $[Z_j-\widehat\delta_j]^2=\min\{\lambda_1^2/4, k^2Z_j^2\}$. Hence $ R(\theta_j, \widehat\theta_{\text{lava}})= E \min\{\lambda_1^2/4, k^2Z_j^2\}+ (1-2k)\sigma^2-2k\sigma^2P(|Z_j|>\lambda_1/(2k)).$ Therefore,
 %$$
 %E\|\widehat\theta_{\text{lava}}-\theta\|_2^2=\sum_{j}R(\theta_j, \widehat\theta_{\text{lava}})= pE \min\{\lambda_1^2/4, k^2Z_j^2\}+ p(1-2k)\sigma^2-2k\sigma^2\sum_{j=1}^pP(|Z_j|>\lambda_1/(2k)),
 %$$
% whose unbiased estimator is given by
  %$$
%\widehat R(\theta)=  \frac{1}{n}\sum_{i=1}^n\sum_{j=1}^p\min\{\lambda_1^2/4, k^2Z_{ij}^2\}+p(1-2k)\sigma^2-2k\sigma^2\frac{1}{n}\sum_{i=1}^n\sum_{j=1}^p1\{|Z_{ij}|>\lambda_1/(2k)\}.
 %$$

 \section{Proofs for Section 3}
 
 \subsection{Proof of Theorem \ref{th3.1}} Let 
$\Q_{\lambda_2} = [X'X + n \lambda_2 I_p]$. Then for any $\delta \in \mathbb{R}^p$
\begin{eqnarray*}
X\{ \widehat \beta(\delta) + \delta\} & = & X \{ \Q_{\lambda_2}^{-1} X' (Y-X\delta) + \delta \} \\
& = & \PR_{\lambda_2} Y + (I_p- \PR_{\lambda_2}) X\delta = \PR_{\lambda_2}Y + \K_{\lambda_2} X\delta.
\end{eqnarray*}
The second claim of the theorem immediately follows from this.

Further, to show the first claim, we can write for any $\delta \in \mathbb{R}^p$,
\begin{eqnarray*}
&& \| Y - X\widehat \beta(\delta) - X \delta \|_2^2 = \| (I_n - \PR_{\lambda_2})(Y- X\delta) \|_2^2 
= \| \K_{\lambda_2} (Y - X\delta) \|_2^2, \\
&& n \lambda_2 \| \widehat \beta (\delta)\|_2^2 =  n \lambda_2 \|  \Q_{\lambda_2}^{-1} X'(Y - X\delta) \|_2^2. 
\end{eqnarray*}
The sum of these terms is equal to
$$
(Y- X\delta)'[\K_{\lambda_2}^2 +  n \lambda_2 X \Q_{\lambda_2}^{-1}\Q_{\lambda_2}^{-1} X'] (Y - X \delta) = \| \K_{\lambda_2}^{1/2} (Y- X \delta) \|_2^2,
$$
where the equality follows from the observation that,  since $\K_{\lambda_2}^2 = I_n  - 2 X \Q_{\lambda_2}^{-1} X' + X \Q_{\lambda_2}^{-1} X'X \Q_{\lambda_2}^{-1}X'$
and $[X'X + n\lambda_2 I_p] \Q_{\lambda_2}^{-1} = I_p$, we have
\begin{eqnarray*}
\K_{\lambda_2}^2 + n \lambda_2 X \Q_{\lambda_2}^{-1}\Q_{\lambda_2}^{-1} X' 
& = & I_n 
- 2 X \Q_{\lambda_2}^{-1} X'
+ X \Q_{\lambda_2}^{-1} [X'X + n \lambda_2 I_{p}] \Q_{\lambda_2}^{-1}X' \\
& = & I_n - X \Q_{\lambda_2}^{-1} X' = I - \Pr_{\lambda_2} =  \K_{\lambda_2}.
\end{eqnarray*}
Therefore, after multiplying by $n $, the profiled objective function in (\ref{lava.p})
can be expressed as:
$$
\| \K^{1/2}_{\lambda_2} (Y- X\delta)\|_2^2 + n \lambda_1 \| \delta\|_1.
$$
This establishes the first claim. \qed

\subsection{Proof of Theorem \ref{th3.2}}Consider  the following lasso problem:
$$
h_{\lambda}(\widetilde y):=\arg\min_{\delta \in \mathbb{R}^p} \left \{\frac{1}{n}\|\widetilde y -\K_{\lambda_2}^{1/2} X\delta\|_2^2+\lambda\|\delta\|_1 \right\}.
$$
Let $g_{\lambda}(\widetilde y,X):=\widetilde X h_{\lambda}(\widetilde y)$, where $\widetilde X := K_{\lambda_2}^{1/2} X$.
By Lemmas 1, 3  and 6 of \cite{tibshirani2012degrees}, $y \mapsto g_{\lambda_1}(y,X)$ is  continuous and almost differentiable, and 
$$
\frac{\partial g_{\lambda_1}(\widetilde y,X)}{\partial \widetilde y}=\widetilde X_{\hat J}(\widetilde X_{\hat J}'\widetilde X_{\hat J})^-\widetilde X_{\hat J}' .%\quad  \widetilde X_{\hat J}=\K_{\lambda_2}^{1/2}\widehat X. 
$$
Then by Theorem \ref{th3.1},  $\widetilde Xd_{\lava}(y,X)= \widetilde Xh_{\lambda_1}(K_{\lambda_2}^{1/2}  y)   =g_{\lambda_1}(K_{\lambda_2}^{1/2}y, X )   $. Therefore, 
\begin{eqnarray*}
\nabla_y\cdot(\K_{\lambda_2} Xd_{\lava}(y,X))&=& \text{tr}\left(\K_{\lambda_2}^{1/2}\frac{\partial g_{\lambda_1}(\K_{\lambda_2}^{1/2} y,X )   }{\partial y}\right)=\tr(\K_{\lambda_2}^{1/2}\widetilde X_{\hat J}(\widetilde X_{\hat J}'\widetilde X_{\hat J})^-\widetilde X_{\hat J}'\K_{\lambda_2}^{1/2}).
\end{eqnarray*}
It follows from  (\ref{eq3.4}) that 
\begin{eqnarray*}\df(\widehat\theta)&=&\tr(\PR_{\lambda_2})+\Ep\tr(\K_{\lambda_2}^{1/2}\widetilde X_{\hat J}(\widetilde X_{\hat J}'\widetilde X_{\hat J})^-\widetilde X_{\hat J}'\K_{\lambda_2}^{1/2})\cr
&=&\tr(\PR_{\lambda_2})+\Ep\tr(\widetilde X_{\hat J}(\widetilde X_{\hat J}'\widetilde X_{\hat J})^-\widetilde X_{\hat J}'(I-\PR_{\lambda_2}))\cr
&=&\tr(\PR_{\lambda_2})+\Ep\tr(\widetilde X_{\hat J}(\widetilde X_{\hat J}'\widetilde X_{\hat J})^-\widetilde X_{\hat J}')-\Ep\tr(\widetilde X_{\hat J}(\widetilde X_{\hat J}'\widetilde X_{\hat J})^-\widetilde X_{\hat J}'\PR_{\lambda_2})\cr
&=&\Ep\rank(\widetilde X_{\hat J})+\Ep\tr(\widetilde \K_{\hat J}\PR_{\lambda_2}). \end{eqnarray*}
\qed

\subsection{Proof of  Lemma \ref{l3.1}}
  Note that $X\widehat\theta_{\text{lava}}+\PR_{\hat J}\widehat U=\PR_{\hat J}Y+\K_{\hat J}X\widehat\theta_{\text{lava}}= \PR_{\hat J}Y+\K_{\hat J}X\widehat\beta+\K_{\hat J}X\widehat\delta$ and  $ X\widetilde\theta_{\text{post-lava}}= \PR_{\hat J}Y+\K_{\hat J}X\widehat\beta$. Hence it suffices to show that $\K_{\hat J}X\widehat\delta=0$. In fact, let $\widehat\delta_{\hat J}$ be the vector of zero components of $\widehat\delta$, then $X\widehat\delta=X_{\hat J}\widehat\delta_{\hat J}$. So $\K_{\hat J}X\widehat\delta=\K_{\hat J}X_{\hat J}\widehat\delta_{\hat J}=0$ since $\K_{\hat J}X_{\hat J}=0.$
 \qed

\subsection{Proof of Theorem \ref{th3.3}}  \textbf{Step 1.}
By (\ref{eq3.4add}),
\begin{eqnarray*} 
\frac{1}{n}\|X\widehat\theta_{\text{lava}}-X\theta_0\|_2^2 & \leq &  \frac{2}{n}\|\K_{\lambda_2}X(\widehat\delta-\delta_0)\|_2^2+\frac{2}{n}\|\textsf{D}_{\text{ridge}}(\lambda_2)\|_2^2 \\
& \leq &  \frac{2}{n}\|\K^{1/2}_{\lambda_2}X(\widehat\delta-\delta_0)\|_2^2 \| \K_{\lambda_2}\| +\frac{2}{n}\|\textsf{D}_{\text{ridge}}(\lambda_2)\|_2^2, 
\end{eqnarray*}
since $ \| \K_{\lambda_2}\| \leq 1$ as shown below. Step 2 provides the bound $ (B_1(\delta_0)  \vee B_2(\beta_0)) \| \K_{\lambda_2}\|$ for the first term,
and Step 3 provides the bound $B_3 + B_4(\beta_0)$ on the second term.

Furthermore, since $X'\K_{\lambda_2}X=n\lambda_2S(S+\lambda_2I)^{-1}$, we have
 $$
 B_2(\beta_0) = \frac{8}{n}\|\widetilde X\beta_0\|_2^2=\frac{8}{n}\beta_0'X'\K_{\lambda_2}X\beta_0=8\lambda_2\beta_0'S(S+\lambda_2I)^{-1}\beta_0.
 $$
 
 Also, to show that $\|\K_{\lambda_2}\|_2\leq 1$, we let $\PR_{\lambda_2}=U_1D_1U_1'$ be the eigen-decomposition of $\PR_{\lambda_2}$, then $\|\K_{\lambda_2}\| =\|U_1(I-D_1)U_1'\|=\|I-D_1\|$. Note that all the nonzero eigenvalues of $D_1$ are the same as those of  $(X'X+n\lambda_2I)^{-1/2}X'X(X'X+n\lambda_2I)^{-1/2}$, and are $\{d_j/(d_j+n\lambda_2), j\leq \min\{n,p\}\}$, where $d_j$ is the $j$th largest eigenvalue of $X'X$. Thus $\|I-D_1\|=\max\{\max_j  n\lambda_2/(d_j+n\lambda_2), 1\}\leq 1.$

 %Let $X=U_1D_1V_1$ be the singular value decomposition of $X$, where $U_1, V_1$ are unitary matrices and the nonzero elements of $D_1$ are denoted by $\{d_1,...,d_r\}$ with $r\leq\min\{n,p\}$. Then $ \PR_{\lambda_2}=U_1D_1(D_1'D_1+n\lambda_2I)^{-1}D_1'U_1' $ and
%$$\K_{\lambda_2}=U_1[I-D_1(D_1'D_1+n\lambda_2I)^{-1}D_1']U_1'.$$ Thus $\|\K_{\lambda_2}\|_2=\|I-D_1(D_1'D_1+n\lambda_2I)^{-1}D_1'\|_2$. Note that $D:=D_1(D_1'D_1+n\lambda_2I)^{-1}D_1'$ is a diagonal matrix, whose nonzero elements are the same as those of \\$(D_1'D_1+n\lambda_2I)^{-1/2}D_1'D_1(D_1'D_1+n\lambda_2I)^{-1/2}$, which are $\{d_i^2/(d_i^2+n\lambda_2): i\leq r\}$. Hence the eigenvalues of $I-D$ are  $\{n\lambda_2/(d_i^2+n\lambda_2): i\leq r\}\cup\{1,...,1\}$, which are bounded by one.

 Combining these bounds yields the result.

\textbf{Step 2.} Here we claim that on the event $\|\frac{2}{n}\widetilde X'\widetilde U\|_{\infty}\leq c^{-1}\lambda_1$, which holds
with probability $1- \alpha$, we have  
 $$
 \frac{1}{n}\|\widetilde X(\widehat\delta-\delta_0)\|_2^2\leq\frac{ 4 \lambda_1^2}{\iota^2(c, \delta_0, \lambda_1, \lambda_2)} \vee \frac{4^2 \|\widetilde X\beta_0\|_2^2}{n} = B_1(\delta_0) \vee  B_{2}(\beta_0).
 $$
  By (\ref{eq3.2}), for any $\delta\in\mathbb{R}^p$, 
 $$
\frac{1}{n}\|\widetilde Y-\widetilde X\widehat\delta\|_2^2+\lambda_1\|\widehat\delta\|_1\leq \frac{1}{n}\|\widetilde Y-\widetilde X\delta_0\|_2^2 +\lambda_1\|\delta_0\|_1.
 $$
 Note that $\widetilde Y=\widetilde X\delta_0+\widetilde U+\widetilde X\beta_0$, which implies the following basic inequality: for $\Delta=\widehat\delta-\delta_0$, on the event 
 $
 \|\frac{2}{n}\widetilde X'\widetilde U\|_{\infty}\leq c^{-1} \lambda_1$, 
\begin{eqnarray*}
 \frac{1}{n}\|\widetilde X\Delta\|_2^2 & \leq & \lambda_1 \Bigg ( \| \delta_0\|_1 - \| \delta_0 + \Delta\|_1 +\left|\frac{2}{n} \Delta \widetilde X'\widetilde U \right| \Bigg ) + 2 \left |\frac{1}{{n}}(\widetilde X\Delta)'(\widetilde X\beta_0) \right |\\
&\leq &  \lambda_1 \Bigg ( \| \delta_0\|_1 - \| \delta_0 + \Delta\|_1 + c^{-1}  \lambda \| \Delta\|_1  \Bigg ) + 2 \left \|\frac{1}{\sqrt{n}}\widetilde X\Delta \right \|_2 \left \|\frac{1}{\sqrt{n}}\widetilde X\beta_0 \right \|_2,
\end{eqnarray*}
or, equivalently,
$$
\frac{1}{n}\|\widetilde X\Delta\|_2^2 
\left(1 -\frac{ 2 \left \|\frac{1}{\sqrt{n}} \widetilde X\beta_0 \right \|_2}
{\left \|\frac{1}{\sqrt{n}}\widetilde X\Delta \right \|_2 } \right)
 \leq \lambda_1 \Big ( \| \delta_0\|_1 - \| \delta_0 + \Delta\|_1 + c^{-1}  \lambda \| \Delta\|_1  \Big ).$$ 
If $ \left \|\frac{1}{\sqrt{n}}\widetilde X\Delta \right \|_2 \leq 4 \left \|\frac{1}{\sqrt{n}} \widetilde X\beta_0 \right \|_2 $,
then we are done. Otherwise we have that
$$
\frac{1}{n}\|\widetilde X\Delta\|_2^2  \leq  2 \lambda_1 \Big ( \| \delta_0\|_1 - \| \delta_0 + \Delta\|_1 + c^{-1}  \| \Delta\|_1  \Big ).
$$
Thus  $\Delta \in \mathcal{R}(c, \delta_0, \lambda_1, \lambda_2)$ and hence by the definition of the design-impact factor
$$
\frac{1}{n}\|\widetilde X\Delta\|_2^2  \leq  2\lambda_1 \frac{ \frac{1}{\sqrt{n}}\|\widetilde X\Delta\|_2 }{\iota(c, \delta_0, \lambda_1, \lambda_2)} \implies \frac{1}{\sqrt{n}}\|\widetilde X\Delta\|_2  \leq   \frac{ 2\lambda_1 }{\iota(c, \delta_0, \lambda_1, \lambda_2)}.
$$

%If $\frac{1}{2}\lambda_1\|\Delta_J\|_1> \|\frac{1}{\sqrt{n}}\widetilde X\Delta\|_2\|\frac{1}{\sqrt{n}}\widetilde X\beta_0\|_2$, we have $\frac{1}{n}\|\widetilde X\Delta\|_2^2\leq \lambda_1\|\Delta_J\|_1$ and $\Delta\in A(\frac{2}{3}, J)$. Therefore, $\frac{1}{n}\|\widetilde X\Delta\|_2^2\leq |J|_0\lambda_1^2/\kappa(\frac{2}{3}, J,\lambda_2)$. Otherwise,  we have $\frac{1}{n}\|\widetilde X\Delta\|_2^2\leq \frac{4}{n}\|\widetilde X\beta_0\|_2^2$. 

Combining the two cases yields the claim.

\textbf{Step 3}. Here we bound
$
\frac{2}{n}\|\textsf{D}_{\text{ridge}}(\lambda_2)\|_2^2.$
We have 
$$
\frac{2}{n}\|\textsf{D}_{\text{ridge}}(\lambda_2)\|_2^2  \leq \frac{4}{n}\|\K_{\lambda_2}X\beta_0\|_2^2+\frac{4}{n}\|\PR_{\lambda_2}U\|_2^2.$$
%Since
%$
%\Ep\frac{4}{n}\|\PR_{\lambda_2}U\|_2^2=\frac{4\sigma_u^2}{n}\tr(\PR_{\lambda_2}^2),
%$
%we have for any $\epsilon>0$, $\frac{4}{n}\|\PR_{\lambda_2}U\|_2^2\leq \frac{4\sigma_u^2}{\epsilon n}\tr(\PR_{\lambda_2}^2)$ with probability at least $1-\epsilon$.  Note that $\PR_{\lambda}^2 = S^2(S+\lambda_2I)^{-2}$.

By \cite{hsu2014random}'s exponential inequality for deviation of quadratic form of sub-Gaussian vectors the following  bound applies with probability $1 -\epsilon$:
\begin{eqnarray*}
\frac{4}{n}\|\PR_{\lambda_2}U\|_2^2 &\leq & \frac{4\sigma_u^2}{n}[\tr(\PR_{\lambda_2}^2) + 2 \sqrt{ \tr(\PR_{\lambda_2}^4)  \log (1/\epsilon)} +  2\|\PR^2_{\lambda_2}\| \log(1/\epsilon)],\\
& \leq &  \frac{4\sigma_u^2}{n}[\tr(\PR_{\lambda_2}^2) + 2 \sqrt{ \tr(\PR_{\lambda_2}^2) \|\PR^2_{\lambda_2}\|  \log (1/\epsilon)} +  2\|\PR^2_{\lambda_2}\| \log(1/\epsilon)],\\
& \leq & \frac{4\sigma_u^2}{n}\left [\sqrt{\tr(\PR_{\lambda_2}^2)} +  \sqrt{2} \sqrt{\|\PR^2_{\lambda_2}}\| \sqrt{\log(1/\epsilon)} \right]^2 = B_3,
\end{eqnarray*}
where the second inequality holds by Von Neumann's theorem ({\cite{horn2012matrix}), and the last inequality is elementary.

  Furthermore, note that $\K_{\lambda_2}X=\lambda_2X(S+\lambda_2I)^{-1}$. Hence
 $$
 B_4(\beta_0) = \frac{4}{n}\|\K_{\lambda_2}X\beta_0\|_2^2=4\lambda_2^2\beta_0'(S+\lambda_2I)^{-1}S(S+\lambda_2I)^{-1}\beta_0=4\beta_0'V_{\lambda_2}\beta_0.  \quad \scriptstyle \blacksquare
 $$

\footnotesize

\bibliographystyle{ims}
\bibliography{liaoBib_newest}

\begin{thebibliography}{38}
\expandafter\ifx\csname natexlab\endcsname\relax\def\natexlab#1{#1}\fi
\expandafter\ifx\csname url\endcsname\relax
  \def\url#1{\texttt{#1}}\fi
\expandafter\ifx\csname urlprefix\endcsname\relax\def\urlprefix{URL }\fi

\bibitem[{Belloni and Chernozhukov(2013)}]{belloni2013least}
\textsc{Belloni, A.} and \textsc{Chernozhukov, V.} (2013).
\newblock Least squares after model selection in high-dimensional sparse
  models.
\newblock \textit{Bernoulli} \textbf{19} 521--547.

\bibitem[{Belloni et~al.(2011)Belloni, Chernozhukov and Wang}]{BCW-Biometrika}
\textsc{Belloni, A.}, \textsc{Chernozhukov, V.} and \textsc{Wang, L.} (2011).
\newblock Square-root lasso: pivotal recovery of sparse signals via conic
  programming.
\newblock \textit{Biometrika} \textbf{98} 791--806.

\bibitem[{Belloni et~al.(2014)Belloni, Chernozhukov and Wang}]{BCW-AOS2014}
\textsc{Belloni, A.}, \textsc{Chernozhukov, V.} and \textsc{Wang, L.} (2014).
\newblock Pivotal estimation via square-root lasso in nonparametric regression.
\newblock \textit{The Annals of Statistics} \textbf{42} 757--788.

\bibitem[{Bickel et~al.(2009)Bickel, Ritov and Tsybakov}]{Bickeletal}
\textsc{Bickel, P.}, \textsc{Ritov, Y.} and \textsc{Tsybakov, A.} (2009).
\newblock Simultaneous analysis of lasso and dantzig selector.
\newblock \textit{The Annals of Statistics} \textbf{37} 1705--1732.

\bibitem[{Bunea et~al.(2007)Bunea, Tsybakov and Wegkamp}]{bunea2007sparsity}
\textsc{Bunea, F.}, \textsc{Tsybakov, A.} and \textsc{Wegkamp, M.} (2007).
\newblock Sparsity oracle inequalities for the lasso.
\newblock \textit{Electronic Journal of Statistics} \textbf{1} 169--194.

\bibitem[{Bunea et~al.(2010)Bunea, Tsybakov, Wegkamp and
  Barbu}]{bunea2010spades}
\textsc{Bunea, F.}, \textsc{Tsybakov, A.~B.}, \textsc{Wegkamp, M.~H.} and
  \textsc{Barbu, A.} (2010).
\newblock Spades and mixture models.
\newblock \textit{The Annals of Statistics} \textbf{38} 2525--2558.

\bibitem[{Candes and Tao(2007)}]{candes2007dantzig}
\textsc{Candes, E.} and \textsc{Tao, T.} (2007).
\newblock The dantzig selector: Statistical estimation when p is much larger
  than n.
\newblock \textit{The Annals of Statistics} \textbf{35} 2313--2351.

\bibitem[{Cand{\`e}s et~al.(2011)Cand{\`e}s, Li, Ma and
  Wright}]{candes2011robust}
\textsc{Cand{\`e}s, E.~J.}, \textsc{Li, X.}, \textsc{Ma, Y.} and
  \textsc{Wright, J.} (2011).
\newblock Robust principal component analysis?
\newblock \textit{Journal of the ACM (JACM)} \textbf{58} 11.

\bibitem[{Chandrasekaran et~al.(2011)Chandrasekaran, Sanghavi, Parrilo and
  Willsky}]{chandrasekaran2011rank}
\textsc{Chandrasekaran, V.}, \textsc{Sanghavi, S.}, \textsc{Parrilo, P.~A.} and
  \textsc{Willsky, A.~S.} (2011).
\newblock Rank-sparsity incoherence for matrix decomposition.
\newblock \textit{SIAM Journal on Optimization} \textbf{21} 572--596.

\bibitem[{Chen and Dalalyan(2012)}]{dalalyan2012fused}
\textsc{Chen, Y.} and \textsc{Dalalyan, A.} (2012).
\newblock Fused sparsity and robust estimation for linear models with unknown
  variance.
\newblock \textit{Advances in Neural Information Processing Systems}
  1259--1267.

\bibitem[{Chernozhukov et~al.(2015)Chernozhukov, Hansen and Liao}]{lava}
\textsc{Chernozhukov, V.}, \textsc{Hansen, C.} and \textsc{Liao, Y.} (2015).
\newblock A lava attack on the recovery of sums of dense and sparse signals.
\newblock Tech. rep., MIT.

\bibitem[{Donoho and Johnstone(1995)}]{donoho1995adapting}
\textsc{Donoho, D.~L.} and \textsc{Johnstone, I.~M.} (1995).
\newblock Adapting to unknown smoothness via wavelet shrinkage.
\newblock \textit{Journal of the American Statistical Association} \textbf{90}
  1200--1224.

\bibitem[{Dossal et~al.(2011)Dossal, Kachour, Fadili, Peyr{\'e} and
  Chesneau}]{dossal2011degrees}
\textsc{Dossal, C.}, \textsc{Kachour, M.}, \textsc{Fadili, J.~M.},
  \textsc{Peyr{\'e}, G.} and \textsc{Chesneau, C.} (2011).
\newblock The degrees of freedom of the lasso for general design matrix.
\newblock \textit{arXiv preprint:1111.1162} .

\bibitem[{Efron(2004)}]{efron2004estimation}
\textsc{Efron, B.} (2004).
\newblock The estimation of prediction error.
\newblock \textit{Journal of the American Statistical Association} \textbf{99}
  619--642.

\bibitem[{Efron et~al.(2004)Efron, Hastie, Johnstone and
  Tibshirani}]{efron2004least}
\textsc{Efron, B.}, \textsc{Hastie, T.}, \textsc{Johnstone, I.} and
  \textsc{Tibshirani, R.} (2004).
\newblock Least angle regression.
\newblock \textit{The Annals of statistics} \textbf{32} 407--499.

\bibitem[{Fan and Li(2001)}]{fan2001variable}
\textsc{Fan, J.} and \textsc{Li, R.} (2001).
\newblock Variable selection via nonconcave penalized likelihood and its oracle
  properties.
\newblock \textit{Journal of the American Statistical Association} \textbf{96}
  1348--1360.

\bibitem[{Fan et~al.(2013)Fan, Liao and Mincheva}]{POET}
\textsc{Fan, J.}, \textsc{Liao, Y.} and \textsc{Mincheva, M.} (2013).
\newblock Large covariance estimation by thresholding principal orthogonal
  complements (with discussion).
\newblock \textit{Journal of the Royal Statistical Society, Series B}
  \textbf{75} 603--680.

\bibitem[{Fan and Lv(2008)}]{fan2008sure}
\textsc{Fan, J.} and \textsc{Lv, J.} (2008).
\newblock Sure independence screening for ultrahigh dimensional feature space.
\newblock \textit{Journal of the Royal Statistical Society: Series B (with
  discussion)} \textbf{70} 849--911.

\bibitem[{Frank and Friedman(1993)}]{frank1993statistical}
\textsc{Frank, L.~E.} and \textsc{Friedman, J.~H.} (1993).
\newblock A statistical view of some chemometrics regression tools.
\newblock \textit{Technometrics} \textbf{35} 109--135.

\bibitem[{Hirano and Porter(2012)}]{hirano2012impossibility}
\textsc{Hirano, K.} and \textsc{Porter, J.~R.} (2012).
\newblock Impossibility results for nondifferentiable functionals.
\newblock \textit{Econometrica} \textbf{80} 1769--1790.

\bibitem[{Horn and Johnson(2012)}]{horn2012matrix}
\textsc{Horn, R.~A.} and \textsc{Johnson, C.~R.} (2012).
\newblock \textit{Matrix analysis}.
\newblock Cambridge university press.

\bibitem[{Hsu et~al.(2014)Hsu, Kakade and Zhang}]{hsu2014random}
\textsc{Hsu, D.}, \textsc{Kakade, S.~M.} and \textsc{Zhang, T.} (2014).
\newblock Random design analysis of ridge regression.
\newblock \textit{Foundations of Computational Mathematics} \textbf{14}
  569--600.

\bibitem[{Klopp et~al.(2014)Klopp, Lounici and Tsybakov}]{klopp2014robust}
\textsc{Klopp, O.}, \textsc{Lounici, K.} and \textsc{Tsybakov, A.~B.} (2014).
\newblock Robust matrix completion.
\newblock \textit{arXiv preprint arXiv:1412.8132} .

\bibitem[{Loh and Wainwright(2013)}]{loh2013regularized}
\textsc{Loh, P.-L.} and \textsc{Wainwright, M.~J.} (2013).
\newblock Regularized m-estimators with nonconvexity: Statistical and
  algorithmic theory for local optima.
\newblock In \textit{Advances in Neural Information Processing Systems}.

\bibitem[{Meinshausen and Yu(2009)}]{meinshausen2009lasso}
\textsc{Meinshausen, N.} and \textsc{Yu, B.} (2009).
\newblock Lasso-type recovery of sparse representations for high-dimensional
  data.
\newblock \textit{The Annals of Statistics} \textbf{37} 246--270.

\bibitem[{Meyer and Woodroofe(2000)}]{meyer2000degrees}
\textsc{Meyer, M.} and \textsc{Woodroofe, M.} (2000).
\newblock On the degrees of freedom in shape-restricted regression.
\newblock \textit{The Annals of Statistics} \textbf{28} 1083--1104.

\bibitem[{Negahban et~al.(2009)Negahban, Yu, Wainwright and
  Ravikumar}]{negahban2009unified}
\textsc{Negahban, S.}, \textsc{Yu, B.}, \textsc{Wainwright, M.~J.} and
  \textsc{Ravikumar, P.~K.} (2009).
\newblock A unified framework for high-dimensional analysis of $ m $-estimators
  with decomposable regularizers.
\newblock In \textit{Advances in Neural Information Processing Systems}.

\bibitem[{Stein(1956)}]{stein1956}
\textsc{Stein, C.~M.} (1956).
\newblock Inadmissibility of the usual estimator for the mean of a multivariate
  distribution.
\newblock \textit{Proc. Third Berkeley Symp. Math. Statist. Prob.}  197–--206.

\bibitem[{Stein(1981)}]{stein1981estimation}
\textsc{Stein, C.~M.} (1981).
\newblock Estimation of the mean of a multivariate normal distribution.
\newblock \textit{The Annals of Statistics} \textbf{9} 1135--1151.

\bibitem[{Sun and Zhang(2012)}]{sun2012scaled}
\textsc{Sun, T.} and \textsc{Zhang, C.-H.} (2012).
\newblock Scaled sparse linear regression.
\newblock \textit{Biometrika} \textbf{99} 879--898.

\bibitem[{Tibshirani(1996)}]{tibshirani96}
\textsc{Tibshirani, R.} (1996).
\newblock Regression shrinkage and selection via the lasso.
\newblock \textit{Journal of the Royal Statistical Society, Series B}
  \textbf{58} 267--288.

\bibitem[{Tibshirani et~al.(2005)Tibshirani, Saunders, Rosset, Zhu and
  Knight}]{tibshirani2005sparsity}
\textsc{Tibshirani, R.}, \textsc{Saunders, M.}, \textsc{Rosset, S.},
  \textsc{Zhu, J.} and \textsc{Knight, K.} (2005).
\newblock Sparsity and smoothness via the fused lasso.
\newblock \textit{Journal of the Royal Statistical Society: Series B
  (Statistical Methodology)} \textbf{67} 91--108.

\bibitem[{Tibshirani and Taylor(2012)}]{tibshirani2012degrees}
\textsc{Tibshirani, R.~J.} and \textsc{Taylor, J.} (2012).
\newblock Degrees of freedom in lasso problems.
\newblock \textit{The Annals of Statistics} \textbf{40} 1198--1232.

\bibitem[{Wainwright(2009)}]{wainwright2009sharp}
\textsc{Wainwright, M.} (2009).
\newblock Sharp thresholds for noisy and high-dimensional recovery of sparsity
  using l1-constrained quadratic programming (lasso).
\newblock \textit{IEEE Transactions on Information Theory} \textbf{55}
  2183--2202.

\bibitem[{Yuan and Lin(2006)}]{yuan2006model}
\textsc{Yuan, M.} and \textsc{Lin, Y.} (2006).
\newblock Model selection and estimation in regression with grouped variables.
\newblock \textit{Journal of the Royal Statistical Society: Series B
  (Statistical Methodology)} \textbf{68} 49--67.

\bibitem[{Zhang(2010)}]{zhang2010nearly}
\textsc{Zhang, C.-H.} (2010).
\newblock Nearly unbiased variable selection under minimax concave penalty.
\newblock \textit{The Annals of Statistics} \textbf{38} 894--942.

\bibitem[{Zhao and Yu(2006)}]{zhao2006model}
\textsc{Zhao, P.} and \textsc{Yu, B.} (2006).
\newblock On model selection consistency of lasso.
\newblock \textit{The Journal of Machine Learning Research} \textbf{7}
  2541--2563.

\bibitem[{Zou and Hastie(2005)}]{zou2005regularization}
\textsc{Zou, H.} and \textsc{Hastie, T.} (2005).
\newblock Regularization and variable selection via the elastic net.
\newblock \textit{Journal of the Royal Statistical Society: Series B
  (Statistical Methodology)} \textbf{67} 301--320.

\end{thebibliography}

\end{document}